\documentclass[11pt, a4paper]{article} 

\usepackage[margin=1.4in]{geometry} 

\usepackage{amsfonts}
\usepackage{amssymb}
\usepackage{amsmath}
\usepackage{amsthm}
\usepackage{dsfont}
\usepackage{graphics}
\usepackage{graphicx}
\usepackage{epsfig}
\usepackage{mathrsfs}
\usepackage[all]{xy}
\usepackage{color}
\usepackage[authoryear,round,comma]{natbib}

\usepackage[title]{appendix}
\usepackage [pagewise] {lineno}

\newtheorem{theorem}{Theorem}[section]
\newtheorem{assumption}[theorem]{Assumption}

\newtheorem{definition}[theorem]{Definition}
\newtheorem{example}[theorem]{Example}

\newtheorem{lemma}[theorem]{Lemma}

\newtheorem{proposition}[theorem]{Proposition}
\newtheorem{remark}[theorem]{Remark}

\numberwithin{equation}{section}
\usepackage{hyperref}
\hypersetup{
colorlinks=true, 
breaklinks=true, 
urlcolor= blue, 
linkcolor= blue, 
citecolor=blue, 
pdftitle={}, 
pdfauthor={}, 
pdfsubject={} 
}

\def\proof{\noindent\textbf{Proof. }}

\def\Pbb{\mathbb{P}}

\def\Rbb{\mathbb{R}}

\def\Bc{\mathcal{B}}

\def\Fc{\mathcal{F}}

\def\Jc{\mathcal{J}}

\def\Rc{\mathcal{R}}

\def\Xc{\mathcal{X}}

\def\Msc{\mathscr{M}}

\def\Qsc{\mathscr{Q}}

\def\Xsc{\mathscr{X}}

\def\var{\mathrm{Var}}

\def\ep{\mathrm{E}}

\DeclareMathOperator*{\essinf}{ess\,inf}
\DeclareMathOperator*{\maxi}{maximize\,}
\DeclareMathOperator*{\mini}{minimize\,}

\def\trieq{\triangleq}

\def\id{\mathbf{1}}
\def\qed{\hbox{ }\hfill {$\Box$}}

\newcommand{\icx}{\succeq_{\mathrm{icx}}}
\newcommand{\icv}{\succeq_{\mathrm{icv}}}
\newcommand{\Qicx}{\mathscr{Q}_\mathrm{icx}}
\newcommand{\Wicx}{\mathscr{W}^\mathrm{icx}}

\newcommand{\Xicx}{\mathscr{X}_\mathrm{icx}}

\usepackage[T3,T1]{fontenc}
\DeclareSymbolFont{tipa}{T3}{cmr}{m}{n}
\DeclareMathAccent{\invbreve}{\mathalpha}{tipa}{16} 

\begin{document}

\title{Benchmark Beating with the Increasing Convex Order\thanks{Supported by the National Key R\&D Program of China (2020YFA0712700) and NSFC (12071146). 
}}

\author{Jianming Xia\thanks{RCSDS, NCMIS, Academy of Mathematics and Systems Science,
Chinese Academy of Sciences, Beijing 100190, China; Email:
xia@amss.ac.cn.}
}

\maketitle

\begin{abstract}
In this paper we model benchmark beating with the increasing convex order (ICX order).  The mean constraint in the mean-variance theory of portfolio selection can be regarded as beating a constant. We then investigate the problem of minimizing the variance of a portfolio with ICX order constraints, based on which we also study the problem of beating-performance-variance efficient portfolios. The optimal and efficient portfolios are all worked out in closed form for complete markets. 

\end{abstract}

\section{Introduction}

Portfolio selection involves a trade-off between return and risk. In contrast to the expected utility theory which models return seeking and risk aversion of an investor with the monotonicity and the concavity of the utility function, the mean-variance theory of \citet{m52} models the return and the risk of a portfolio with its mean and variance. Thereby the trade-off between the return and the risk is explicitly formulated in Markowitz's mean-variance theory. 

In spite of the popularity of the mean-variance theory in both academia and industry, variance has been frequently criticized as a risk measure for its main shortcoming: it treats the volatile positive returns as a part of risk. While the return is measured by mean, some alternative risk measures have been proposed to replace variance in portfolio selection, e.g., 
value-at-risk (VaR) (\citet{Campbell2001}, \citet{Alex2002}), conditional VaR (\citet{RU2000}), weighted VaR (\citet{HeJinZhou2015}), entropic VaR (\citet{Ahmadi2019}), and expectile (\citet{WagUry2019}, \citet{LSW2021}); see \citet{HeJinZhou2015} and \citet{LSW2021} for reviews on the development along this line. 

Instead of the risk aversion, we focus on the other side---the return seeking---of portfolio selection in this paper. In the mean-variance theory, the return of a portfolio is measured by its mean,  which, unlike variance as a risk measure,  has been seldom criticized. Moreover, a mean-variance efficient portfolio can be found out by minimizing the variance under the constraint that the return is no less than a benchmark level, which is modeled by a constant. 
We consider an extension of the mean constraint such that the benchmark is modeled by a random variable, which is much more flexible in modeling return seeking than a constant. To this end, we use a stochastic order.  

Given a random variable $X_0$, which is called a benchmark, we say a random variable $X$ beats the benchmark $X_0$ and write $X\icx X_0$ if $X$ dominates $X_0$ in the increasing convex order (ICX order), i.e.,
\begin{equation}\label{eq:efx}
\ep[f(X)]\ge\ep[f(X_0)]
\end{equation}
for all increasing convex functions $f$ provided that the expectations are well defined. It is well known that $X\icx X_0$ if and only if
\begin{equation}\label{eq:icx:qxq0}
\int_t^1 Q_X(s)ds\ge\int_t^1 Q_{X_0}(s)ds,\quad t\in(0,1),
\end{equation}
where $Q_X$ and $Q_{X_0}$ are quantile functions of $X$ and $X_0$. 

The financial meaning of \eqref{eq:icx:qxq0} is clearer when $X_0$ is discretely distributed:
\begin{equation}\label{eq:X0:ain}
\Pbb(X_0=a_i)=p_i,\quad 1\le i\le n,
\end{equation}
where $n\ge1$, $a_1<a_2<\dots<a_n$, $p_i\in(0,1]$ for all $i=1,\dots,n$, and $\sum_{i=1}^np_i=1$.
Let 
\[q_0=0,\quad q_i=p_1+\dots+p_i, \quad i=1,\dots,n-1.\] 
It is easy to see that both sizes of \eqref{eq:icx:qxq0} are concave w.r.t. $t$.  Moreover,  the right hand size of \eqref{eq:icx:qxq0} is piecewise  linear in the special case \eqref{eq:X0:ain}.  
Therefore,  \eqref{eq:icx:qxq0} is equivalent to
\begin{equation}\label{ineq:epeu}
\int_{q_i}^1 Q_X(s)ds\ge\int_{q_i}^1Q_{X_0}ds,\quad i=0,\dots,n-1.
\end{equation}
In \eqref{ineq:epeu}, the constraint with $i=0$ amounts to $\ep[X]\ge\ep[X_0]$.
For each $i=1,\dots,n-1$, $\int_{q_i}^1 Q_X(s)ds$ is the average gain of $X$ in the best $(100\times q_i)\%$ of cases and ignores the values of $X$ in the worst $(100\times (1-q_i))\%$ of cases. It is called \emph{expected upside,  expected gain or tail gain}; see \citet{Bacon2013} for example.\footnote{In contrast, the expected shortfall or conditional value at risk with confidence level $\alpha$ is the average loss of $X$ in the worst $100\times\alpha\%$ of cases. The ratio of expected upside and expected shortfall gives the Rachev ratio in finance, see 
\citet{Biglova2004}.} Therefore \eqref{ineq:epeu} consists of a constraint on the mean of $X$ and some other constraints on the expected upsides of $X$. The ICX order imposes not only the usual mean constraint as in the mean-variance theory but also constraints on the right tail. The constraints on the right tail 
makes the ICX order appropriate for modeling benchmark beating (right-tail reward control), since by beating a benchmark we expect that $X$ is more profitable than the benchmark $X_0$. 

For a constant $z$, $X\icx z$ if and only if $\ep[X]\ge z$.  Therefore, the mean constraint $\ep[X]\ge z$ can be regarded as beating constant $z$.  
The aim of the present paper is to extend the mean constraint in the mean-variance theory to the ICX order constraint. More precisely, we consider the following problem: 
\begin{equation*}
\mini_{X\in\Xsc}\var[X]\quad\text{subject to }X\icx X_0,
\end{equation*}
where $\Xsc$ is a feasible set. 

As initiating work, we study the problem of portfolio selection with ICX order constraint within a complete
market. The original problem is generally non-convex since the set $\{X\mid X\icx X_0\}$ is generally non-convex. But the set $\{Q_X\mid X\icx X_0\}$ is always convex. Then by the recently developed quantile formulation (e.g., \citet{Schied2004}, \citet{CarlierDana2006}, \citet{JinZhou2008}, and \citet{HeZhou2011}), 
we transform the original non-convex problem into a convex one and finally provide the optimal solution in closed form.
Moreover, for a payoff $X$, let its performance of beating $X_0$ be defined as 
\[
\psi(X)\trieq\sup\{m\in\Rbb\mid X-m\icx X_0\}\quad\text{with }\sup\emptyset=-\infty.
\]
A payoff $X\in \Xsc$ is called beating-performance-variance (BPV) efficient in $\Xsc$ if there is no $Y\in\Xsc$
such that 
\[
\psi(Y)\geq\psi(X)\quad\text{and}\quad\var[Y]\leq\var[X]
\]
with at least one inequality holding strictly. 
We then also investigate the problem of BPV efficient portfolio and the BPV efficient portfolios are worked out in closed form.  
In the special case $X_0=0$, $\Psi(X)=\ep[X]$ and hence BPV efficiency reduces to mean-variance efficiency. 

The ICX order is closely related to another stochastic order, the increasing concave order (ICV order), which is also called the second-order stochastic dominance (SSD).  
A random random variable $X$ dominates $X_0$ in the ICV order and write $X\icv X_0$ if 
\eqref{eq:efx} holds for all increasing concave functions $f$ provided that the expectations are well defined. It is well known that $X\icv X_0$ if and only if
$$\int_0^t Q_X(s)ds\ge\int_0^t Q_{X_0}(s)ds,\quad t\in(0,1).$$
The financial meaning of ICX order is significantly different from that of ICV order. Actually, for a constant $z$, $X\icx z$ amounts to $\ep[X]\ge z$, whereas $X\icv z$ amounts to $X\ge z$ a.s. Moreover, by \citet[Proposition 3.2]{DenRus2003}, for a discretely distributed benchmark $X_0$ satisfying \eqref{eq:X0:ain},   $X\icv X_0$ if and only if 
\begin{equation}\label{ineq:ssd:m}
X\ge a_1 \text{ a.s. and } \ep[(a_i-X)^+]\le \ep[(a_i-X_0)^+],\ i=2,\dots,n.
\end{equation}
The constraints in \eqref{ineq:ssd:m} pay attention to the left-tail risk of $X$ measured by $\essinf X$ and $\ep[(a_i-X)^+$ ($i=2,\dots,n$). 
Therefore, the ICV order is appropriate for modeling risk control (left-tail risk control). In a series of papers, 
\citet{DenRus2003, DenRus2004a, DenRus2006a, DenRus2006b} introduced a stochastic optimization
model with ICV order constraints. As an application, they
also investigated the problem of portfolio selection with risk
control by ICV order instead of variance. For further developments along this line, see
\citet{RudRus2008}, \citet{Fabian2011} and \citet{DenMW2016}. In particular,  \citet{WangXia2021} investigate the problem of expected utility maximization with the ICV order constraint. 

The rest of this paper is organized as follows. Sections \ref{sec:probform} and \ref{sec:QF} introduce the basic model and the quantile formulation, respectively. Section \ref{sec:duality} establishes the duality between the primal and the dual problems. 
Section \ref{sec:dual:o} investigates the dual problem and gives the dual optimizer in closed form. 
Section \ref{sec:po} studies the primal optimizer. 
Section \ref{sec:special} presents the optimal solution for the special case of two-point distributed benchmark. 
Section \ref{sec:disc} discusses some problems including BPV efficient portfolios, multi-benchmark beating, and mean-variance portfolio selection with ICX order constraint,  
The Appendix collects some proofs.

\section{Problem Formulation}\label{sec:probform}

Consider a complete nonatomic probability space $\left(\Omega,\Fc,\Pbb\right)$.
Let 
$L^{2}$ (resp., $L^\infty$) denote all squarely integrable (resp., essentially bounded) $\mathcal{F}$-measurable
random variables. 
The (upper) quantile function $Q_X$ of a random variable $X$ is defined by 
\[Q_X(t)\trieq \inf\{x\in\Rbb\mid\Pbb(X\leq x)>t\},\quad t\in [0,1).\]
By convention, we extend $Q_X$ by letting $Q_X(1)\trieq Q_X(1-)=\lim_{t\uparrow 1}Q_X(t)$ and $Q_X(0-)\trieq0$ if needed. For more details about quantile functions, see, e.g.,  \citet[Appendix A.3]{FS2016}.

Let $X$ and $Y$ be two random variables. 
We say that $X$ dominates $Y$ in the \textit{increasing convex order (ICX order)} and write $X\icx Y$ if  
$\ep[f(X)]\ge\ep[f(Y)]$ for all increasing convex functions $f$ provided that the expectations exist. It is well known that\footnote{Throughout the paper, for $t_1<t_2$, we use $\int_{t_1}^{t_2}$ to denote the integration on the open interval $(t_1,t_2)$, that is,
$\int_{t_1}^{t_2}=\int_{(t_1,t_2)}$.}
\begin{align}\label{eq:icx:eqv}
X\icx Y 
\Longleftrightarrow&\int_t^1Q_X(s)ds\ge\int_t^1Q_Y(s)ds \text{ for all }t\in[0,1]\\
\Longleftrightarrow&\int_{[0,1]}Q_X(s)dw(s)\ge\int_{[0,1]}Q_Y(s)dw(s)\text{ for all }w\in\Wicx,\nonumber
\end{align}
where 
$$\Wicx\trieq\{w:[0,1]\to[0,\infty)\mid w \text{ is increasing and convex with } w(0)=0\}.$$
By convention, we set $w(0-)\trieq0$ for $w\in\Wicx$.
We can identify each $w\in\Wicx$ as a finite measure on the measurable space $\left([0,1], \Bc_{[0,1]}\right)$, where $\Bc_{[0,1]}$ denotes all Borel subsets of $[0,1]$. For each $w\in\Wicx$, the measure of $\{0\}$ is $0$ and the measure of $\{1\}$ is $w(1)-w(1-)$. Furthermore, for each $w\in\Wicx$ and $X\in L^2$, by the convexity of $w$ and $Q_X\in L^2([0,1))$, we know that $\int_{[0,1]}Q_X(s)dw(s)=\int_{(0,1]}Q_X(s)dw(s)$ is well defined and
$\int_{[0,1]}Q_X(s)dw(s)>-\infty$. 
For more details about the increasing convex order, see, e.g., \citet[Chapter 4]{SS2007}. 

Consider an arbitrage-free market. 
Assume that the market is complete and has a unique 
stochastic discount factor (SDF) $\rho\in L^2$ with $\Pbb(\rho>0)=1$ and variance $\var[\rho]>0$. 

Let us consider an investor with initial capital $\mathit{x}$, who
wants to minimize the risk measured by the payoff's variance $\var[X]$ while beating a benchmark payoff $X_0\in L^\infty$ in the sense of the ICX order.
The problem of the investor is thus  
\begin{equation}\label{opt:vm:x}
\begin{split}
&\mini_{X\in L^2} \var[X]\\
&\text{subject to } \ep[\rho X]\le x,\; X\icx X_0.
 \end{split}
\end{equation}

\begin{remark}\label{rmk:z} 
Obviously, for every $z\in\Rbb$,
$X\icx z$ if and only if $\ep[X]\ge z$. Then for $X_0=z$ a.s., problem \eqref{opt:vm:x} reduces to 
\begin{equation*}\begin{split}
&\mini_{X\in L^2} \var[X]\\
&\text{\rm subject to } \ep[\rho X]\le x,\; \ep[X]\ge z,
\end{split}
\end{equation*}
which arises from the classical problem of mean-variance portfolio selection. Therefore, the classical mean-variance portfolio selection is to minimize variance while beating a constant $z$. 
\end{remark}

Let 
\[
\Xicx(x,X_0)\triangleq\{X\in L^{2}\:\vert\:\ep[\rho X]\leq x\text{ and }X\icx X_{0}\}.
\]
A payoff $X$ is called \textit{variance-minimal} in $\Xicx(x,X_0)$ if it solves problem \eqref{opt:vm:x}.

From \eqref{eq:icx:eqv}, the advantage of ICX order is obvious: it pays more attention to the right tail (the gain part) than to the left tail (the loss part) of a random variable. This advantage makes the ICX order appropriate for modeling benchmark beating, since by beating a benchmark we expect that the payoff $X$ has more profit than the benchmark $X_0$. The mean, however, pays equal attention to the right and the left tails.

\section{Quantile Formulation}\label{sec:QF}

In general, the set $\Xicx(x,X_0)$ is not convex, which leads to a difficulty of the problem. The appropriate
technique for overcoming this difficulty is the well-developed ``quantile
formulation'':  changing the decision variable of the problem from the random
variable $X$ to its quantile function $Q_X$; see \citet{Schied2004}, \citet{CarlierDana2006}, \citet{JinZhou2008}, and \citet{HeZhou2011}. This formulation recovers the
implicit convexity (in terms of quantile functions) of $\Xicx(x,X_0)$.

Let
\[
\mathscr{Q}\triangleq\left\{ Q:[0,1)\rightarrow\mathbb{R}\:\bigg\vert\:Q\:\text{is increasing, right-continuous and}\:\int_{0}^{1}Q^{2}(s)ds<\infty\right\} .
\]
Obviously, $\mathscr{Q}$ is the set of quantile functions of random
variables $\mathit{X}\in L^{2}$, that is,
\[
\mathscr{Q}=\{Q_{X}\:\vert\:X\in L^{2}\}.
\]
For any $Q_1, Q_2\in\Qsc$, we write $Q_1\icx Q_2$ if 
$$\int_t^1Q_1(s)ds\ge\int_t^1Q_2(s)ds\quad \text{ for all }t\in[0,1].$$
Obviously, 
$$Q_1\icx Q_2
\Longleftrightarrow \int_{[0,1]}Q_1(s)dw(s)\ge\int_{[0,1]} Q_2(s)dw(s)\text{ for all }w\in\Wicx.
$$
For any $Q\in\Qsc$, we have
\begin{align}\label{eq:QQ0w}
\begin{aligned}
&Q\icx Q_0\Longleftrightarrow \inf_{w\in\Wicx}\left(\int_{[0,1]}(Q(s)-Q_0(s))dw(s)\right)=0,\\
&Q\not\icx Q_0\Longleftrightarrow \inf_{w\in\Wicx}\left(\int_{[0,1]}(Q(s)-Q_0(s))dw(s)\right)=-\infty.
\end{aligned}
\end{align}

For notational simplicity, we write $\mathit{Q_{\mathrm{0}}}$ instead of $\mathit{Q_{X_{\mathrm{0}}}}$. 
Let
$$
\Qicx(x,Q_0)  \triangleq\left\{ Q\in\mathscr{Q}\:\Big|\:\int_{0}^{1}Q(s)Q_{\rho}(1-s)ds\leq x\:\mathrm{and}\:Q\icx Q_0\right\} .
$$
Obviously, the set $\Qicx(x,Q_0)$ is convex and closed in $L^2([0,1))$.  

With abuse of notation, we let 
$$\ep[Q]\trieq\int_0^1Q(s)ds\ \mbox{ and }\var[Q]\triangleq\int_{0}^{1}Q^{2}(s)ds-\left(\int_{0}^{1}Q(s)ds\right)^{2},\quad Q\in\mathscr{Q}.$$

A quantile function $Q\in\Qicx(x,Q_0)$
is called \textit{variance-minimal} in $\Qicx(x,Q_0)$ if it solves the following problem:
\begin{equation}\label{opt:vm:Q}
\begin{split}
&\mini_{Q\in\Qsc} \var[Q]\\
&\text{subject to } \int_{0}^{1}Q(s)Q_{\rho}(1-s)ds\leq x,\; Q\icx Q_0.
\end{split}
\end{equation}
Let $v^\circ(x)$ denote the minimal value of problem \eqref{opt:vm:Q}. Obviously, $v^\circ$ is convex.

The next proposition shows that finding variance-minimal payoffs is equivalent to finding  
variance-minimal quantile functions.

\begin{proposition}\label{prop:X:Q}
If $X$ is variance-minimal in $\Xicx(x,X_0)$, 
then so is $\mathit{Q}_{X}$ in $\Qicx(x,Q_0)$. Conversely,
if $\mathit{Q}$ is variance-minimal in $\Qicx(x,Q_0)$, then so is $\mathit{X}= Q(1-\xi)$ in $\Xicx(x,X_0)$, 
where\footnote{For the existence of such a $\xi$, see, e.g.,  \citet[Lemma  A.28]{FS2016}. Moreover, let $\xi$ be a random variable uniformly distributed on $(0,1)$. Then by \citet[Theorem 5]{Xu2014}, $\xi\in\Xi$ if and only if $(\xi,\rho)$ is comonotonic.}
\[
\xi\in\Xi\triangleq\{\xi\:\vert\:\xi\text{ is uniformly distributed on }(0,1)\text{ and }\rho=Q_{\rho}(\xi)\text{ a.s.}\}.
\]
\end{proposition}

\begin{proof}  It can be proved by the standard arguments in the aforementioned quantile formulation literature.\qed
\end{proof}

Hereafter, we focus on studying 
variance-minimal quantile functions.

\section{Duality}\label{sec:duality}

The following theorem implies that $\Qicx(x,Q_0)\neq\emptyset$ for any $x\in\Rbb$.

\begin{theorem}\label{thm:emp}  
Assume that $X_0\in L^\infty$. Then
$\inf_{Q\icx Q_0}\int_0^1Q(s)Q_\rho(1-s)ds=-\infty$.
\end{theorem}

\proof See Appendix \ref{app:proof:thm:emp}.\qed

The following proposition shows that the problem is trivial when $Q_0(1)\ep[\rho]\le x$.

\begin{proposition}\label{prop:trivial}
Assume that $X_0\in L^\infty$.  If $Q_0(1)\ep[\rho]\le x$, then $v^\circ(x)=0$ and the variance-minimal quantile functions in $\Qicx(x,Q_0)$ are constants in the interval $\left[Q_0(1), {x\over\ep[\rho]}\right]$.
\end{proposition}

\proof  Assume $Q_0(1)\ep[\rho]\le x$. In this case, for every $c\in \left[Q_0(1), {x\over\ep[\rho]}\right]$, we have that $c\in\Qicx(x,Q_0)$, which implies that $v^\circ(x)=0$ and $c$ is variance-minimal in $\Qicx(x,Q_0)$. On the other hand, if $Q^\circ$ is variance-minimal in $\Qicx(x,Q_0)$, then $\var[Q^\circ]=0$ and hence $Q^\circ\equiv c_0$ for some $c_0\in\Rbb$. Then by $c_0\in\Qicx(x,Q_0)$, we have $c_0\in \left[Q_0(1), {x\over\ep[\rho]}\right]$.
 \qed
 
 Hereafter, we always make the following assumption unless otherwise stated. 
 
 \begin{assumption}\label{ass:X0x}
 $X_0\in L^\infty$ and $Q_0(1)\ep[\rho]>x$.
 \end{assumption}

The next theorem establishes the existence and uniqueness of the variance-minimal solution. 

 \begin{theorem}\label{thm:unique}
 Under Assumption \ref{ass:X0x},  $v^\circ(x)>0$ and there exists a unique variance-minimal quantile function in $\Qicx(x,Q_0)$.
 \end{theorem}
 
 \proof See Appendix \ref{app:thm:unique}. \qed

It is well known that $\var[Q]$ is convex in $Q$. 
 By Theorem \ref{thm:emp}, there exists some $Q\in\Qsc$ such that $Q\icx Q_0$ and $\int_0^1Q(s)Q_\rho(1-s)ds<x$. Therefore, the Slater condition holds for the budget constraint in \eqref{opt:vm:Q}. 
 Then by the standard results in convex optimization theory (see, e.g., \citet[Sections 8.3--8.5]{L1969}),  
 we have the following proposition.
 
 \begin{proposition}\label{prop:existence}  
 Assume that $X_0\in L^\infty$. Then for every $x\in\Rbb$, 
 $Q^\circ$ solves problem \eqref{opt:vm:Q} if and only if there exists some $\lambda^\circ\ge0$ such that $Q^\circ$ solves the following problem
\begin{equation*}
\begin{split}
&\mini_{Q\in\Qsc}\quad \var[Q]+\lambda^\circ\int_0^1Q(s)Q_\rho(1-s)ds\\
& \text{\rm subject to }\quad Q\icx Q_0
\end{split}
\end{equation*}
 and satisfies 
 $$\lambda^\circ\left(\int_0^1Q^\circ(s)Q_\rho(1-s)ds-x\right)=0.$$
If that is the case, then $-\lambda^\circ\in\partial v^\circ(x)$, where
$\partial v^\circ$ denotes the subdifferential of $v^\circ$.
 \end{proposition}

 \begin{remark}\label{rmk: lambdao}
 If $Q_0(1)\ep[\rho]>x$, then $\lambda^\circ>0$.  Actually, suppose on the contrary that $\lambda^\circ=0$. Then $Q^\circ$ solves the problem 
\begin{equation*}
\begin{split}
&\mini_{Q\in\Qsc}\quad \var[Q]\\
& \text{\rm subject to }\quad Q\icx Q_0.
\end{split}
\end{equation*} 
Obviously, every constant $c\ge Q_0(1)$ also solves the above problem. Therefore,  
$v^\circ(x)=\var[Q^\circ]=0$, which is impossible by Theorem \ref{thm:unique}.
\end{remark}

 We now consider, for any fixed $\lambda>0$, the following problem:
 \begin{equation}\label{opt:var:lambda}
\begin{split}
&\mini_{Q\in\Qsc}\quad \var[Q]+\lambda\int_0^1Q(s)Q_\rho(1-s)ds\\
& \text{subject to }\quad Q\icx Q_0,
\end{split}
\end{equation}
which is equivalent to
 \begin{equation}\label{opt:var:lambda:beta}
\begin{split}
&\mini_{Q\in\Qsc}\quad \min_{\beta\in\Rbb}\int_0^1(Q(s)-\beta)^2ds+\lambda\int_0^1Q(s)Q_\rho(1-s)ds\\
& \text{subject to }\quad Q\icx Q_0.
\end{split}
\end{equation}

To solve problem \eqref{opt:var:lambda} or \eqref{opt:var:lambda:beta}, we first consider, for any fixed $\lambda>0$ and $\beta\in\Rbb$, the following problem:
 \begin{equation}\label{opt:lambda}
\begin{split}
&\mini_{Q\in\Qsc}\quad \int_0^1(Q(s)-\beta)^2ds+\lambda\int_0^1Q(s)Q_\rho(1-s)ds\\
& \text{subject to }\quad Q\icx Q_0.
\end{split}
\end{equation}

Let $v(\beta,\lambda)$ denote the optimal value of problem \eqref{opt:lambda}. 
Let
\begin{align*}
L(Q,w;\beta,\lambda)\trieq &\int_0^1(Q(s)-\beta)^2ds+\lambda\int_0^1Q(s)Q_\rho(1-s)ds\\
&-\left(\int_{[0,1]}Q(s)dw(s)-\int_{[0,1]}Q_0(s)dw(s)\right),\quad
Q\in\Qsc,\; w\in\Wicx.
\end{align*}
In view of \eqref{eq:QQ0w}, we know that 
\begin{align*}
v(\beta,\lambda)=\inf_{Q\in\Qsc}\sup_{w\in\Wicx}L(Q,w;\beta,\lambda).
\end{align*}
Moreover, we have the following proposition.

\begin{proposition}\label{prop:saddle:L}
Assume that $X_0\in L^\infty$. Let $Q^*\in\Qsc$. Then
$Q^*$ solves problem \eqref{opt:lambda} if and only if there exists some $w^*\in\Wicx$ such that $(Q^*,w^*)$ is a saddle point of $L(\cdot\;, \cdot\;; \beta,\lambda)$ (with respect to minimizing in $Q$ and maximizing in $w$). 
\end{proposition}

\proof See Appendix \ref{app:proof:saddle:L}.\qed

Similarly to Lemma \ref{lma:beta}(a), problem \eqref{opt:lambda} has a unique optimal solution $Q^*_{\beta,\lambda}$. Then by Proposition \ref{prop:saddle:L}, there exists some $w^*_{\beta,\lambda}\in\Wicx$ such that $(Q^*_{\beta,\lambda},w^*_{\beta,\lambda})$ is a saddle point of $L(\cdot\;, \cdot\;; \beta,\lambda)$. 
In this case, $w^*_{\beta,\lambda}$ is an optimal solution to the following dual optimization problem:
\begin{equation}\label{opt:dual:w}
\maxi_{w\in\Wicx}\inf_{Q\in \Qsc} L(Q,w;\beta,\lambda),
\end{equation}
which will be discussed in the next section.

\section{Dual Optimization}\label{sec:dual:o}

Note that each $w\in\Wicx$ is continuous on $[0,1)$ and can be discontinuous at the right end-point $1$.
The next lemma shows that the dual optimizer is continuous on $[0,1]$.

\begin{lemma}\label{lma:w(1)}
Assume that $X_0\in L^\infty$.
Let $w\in\Wicx$. If $w(1)>w(1-)$, then $\inf_{Q\in \Qsc} L(Q,w;\beta,\lambda)=-\infty$. 
\end{lemma}

\proof See Appendix \ref{app:proof:w(1)}. \qed.

The next lemma shows that $w$ solves the dual optimization problem only if $w^\prime\in L^2([0,1))$.

\begin{lemma}\label{lma:inner:infinite}
Assume that $X_0\in L^\infty$.
Let $w\in\Wicx$. If $w(1)=w(1-)$ and $w^\prime\notin L^2([0,1))$, then $\inf_{Q\in \Qsc} L(Q,w;\beta,\lambda)=-\infty$. 
\end{lemma}

\proof See Appendix \ref{app:prrof:inner:infinite}. \qed

Let
\begin{align*}
\Wicx_{c,2}\trieq&\{w\in\Wicx \mid w\text{ is continuous and } w^\prime\in L^2([0,1))\}\\
=&\left\{w:[0,1]\to[0,\infty)\;\left|\;
\begin{aligned}
& w \text{ is increasing, continuous, and convex,}\\
& w^\prime\in L^2([0,1)), w(0)=0
 \end{aligned}
 \right.\right\}.
\end{align*}
Then by Lemmas \ref{lma:w(1)} and \ref{lma:inner:infinite}, the dual optimization problem \eqref{opt:dual:w} is equivalent to 
\begin{equation}\label{opt:dual:wc2}
\maxi_{w\in\Wicx_{c,2}}\inf_{Q\in \Qsc} L(Q,w;\beta,\lambda).
\end{equation}

To solve problem \eqref{opt:dual:wc2}, we first consider, for any given $w\in\Wicx_{c,2}$, the inner optimization problem
\begin{equation}\label{opt:dual:Q:inner}
\mini_{Q\in \Qsc} L(Q,w;\beta,\lambda).
\end{equation}
The next lemma completely solves the inner optimization problem for $w\in\Wicx_{c,2}$.

\begin{lemma}\label{lma:inner}
Assume that $X_0\in L^\infty$.
Let $w\in\Wicx_{c,2}$. Then the optimal solution to problem \eqref{opt:dual:Q:inner} is given by
\begin{equation}\label{eq:Q:inner}
Q(s)=\beta+{w^\prime(s)-\lambda Q_\rho((1-s)-)\over2},\quad s\in[0,1).
\end{equation}
\end{lemma}

\proof 
For every $w\in\Wicx_{c,2}$ and $Q\in\Qsc$, we have
\begin{align*}
 L(Q,w;\beta,\lambda)=&\int_0^1(Q(s)-\beta)^2ds+\lambda\int_0^1Q(s)Q_\rho(1-s)ds\\
&-\int_0^1Q(s)w^\prime(s)ds+\int_0^1Q_0(s)w^\prime(s)ds.
\end{align*}
The pointwise optimizer satisfies the following first-order condition:
$$2(Q(s)-\beta)+\lambda Q_\rho((1-s)-)-w^\prime(s)=0,\quad s\in[0,1),$$
which is equivalent to \eqref{eq:Q:inner}. 
Obviously,  the pointwise optimizer is a quantile function in $\Qsc$ and hence is a true optimal solution to problem \eqref{opt:dual:Q:inner}.\qed
 
For every $w\in\Wicx_{c,2}$,  by Lemma  \ref{lma:inner} and
by plugging \eqref{eq:Q:inner} into $L(Q,w;\beta,\lambda)$, we know that
\begin{align*}
&\min_{Q\in \Qsc} L(Q,w;\beta,\lambda)\\
=&-{1\over4}\int_0^1(w^\prime(s)-\lambda Q_\rho(1-s))^2ds+\int_0^1(Q_0(s)-\beta)w^\prime(s)ds+\lambda\beta\ep[\rho]\\
=&\int_0^1\left(-{1\over4}(w^\prime(s))^2+\left({1\over2}\lambda Q_\rho(1-s)+Q_0(s)-\beta\right)w^\prime(s)\right)ds-{1\over4}\lambda^2\ep[\rho^2]+\lambda\beta\ep[\rho].
\end{align*}
Therefore, the dual optimization problem \eqref{opt:dual:w} reduces to
\begin{equation*}
\maxi_{w\in\Wicx_{c,2}} \int_0^1 \left(-{1\over2}(w^\prime(s))^2-\left(2\beta-\lambda Q_\rho(1-s)-2Q_0(s)\right)w^\prime(s)\right)ds,
\end{equation*}
or, equivalently,
\begin{equation}\label{opt:dual:wc3}
\maxi_{w\in\Wicx_{c,2}} \int_0^1-{1\over2}(w^\prime(s))^2ds-\int_0^1w^\prime(s)dH_{\beta,\lambda}(s),
\end{equation}
where
\begin{equation}\label{eq:def:H}
H_{\beta,\lambda}(s)\trieq \int_0^s(2\beta-\lambda Q_\rho(1-t)-2Q_0(t))dt,\quad s\in[0,1].
\end{equation}
Problem \eqref{opt:dual:wc3} can be reformulated as 
\begin{equation}\label{opt:dual:G}
\begin{split}
&\maxi_{G\in\Qsc} \int_0^1-{1\over2}(G(s))^2ds-\int_0^1G(s)dH_{\beta,\lambda}(s)\\
&\text{subject to }\quad G(s)\ge0, s\in[0,1).
\end{split}
\end{equation}

Problem \eqref{opt:dual:G} is similar to a problem arising in rank-dependent utility maximization. If $H_{\beta,\lambda}$ is concave, then its right derivative
$H_{\beta,\lambda}^\prime$ is decreasing. In this case, the solution
to problem \eqref{opt:dual:G} is given by the pointwise
optimizer $G^*(s)=(-H_{\beta,\lambda}^\prime(s))^+$. In general, $H_{\beta,\lambda}$
is potentially non-concave, and therefore $(-H_{\beta,\lambda}^\prime(s))^+$ is
potentially not increasing, which violates the constraint $G\in\Qsc$. Thus, the problem cannot be solved straightforwardly by 
 pointwise optimization in general.
To overcome such an obstacle, in the context of rank-dependent utility maximization and by calculus of variations, \citet{XiaZhou2016} showed that the solution is given explicitly by the concave envelope. Then \citet{Xu2016} provided the concave envelope relaxation approach to the solution, which is more straightforward. For more applications of concave envelope relaxation, see
\citet{Rogers2009},  \citet{Wei2018}, and \citet{WangXia2021}.

Let $\invbreve H_{\beta,\lambda}$ denotes the concave envelope of $H_{\beta,\lambda}$, that is, $\invbreve
H_{\beta,\lambda}$ is the smallest concave function on $[0,1]$ that is no less
than $H_{\beta,\lambda}$:
$$\invbreve H_{\beta,\lambda}(s)\trieq\inf\{H(s)\,|\, H \mbox{ is concave
and }H\ge H_{\beta,\lambda}\mbox{ on }[0,1]\},\quad s\in[0,1].$$ 
By Lemma \ref{lam:cv:envelope}, 
$$\invbreve H_{\beta,\lambda}(0)=H_{\beta,\lambda}(0)=0,\ \invbreve H_{\beta,\lambda}(1)=H_{\beta,\lambda}(1)=2\beta-\lambda \ep[\rho]-2\ep[X_0],\ \invbreve H_{\beta,\lambda}\ge H,$$
$\invbreve H_{\beta,\lambda}$ is continuous on $[0,1]$, and the right derivative $\invbreve H^\prime_{\beta,\lambda}$ is flat on $[\invbreve H_{\beta,\lambda}>H_{\beta,\lambda}]$.

We replace $H_{\beta,\lambda}$ in \eqref{opt:dual:G} with its concave
envelope $\invbreve H_{\beta,\lambda}$ and consider the problem
\begin{equation}\label{opt:dual:G:hat}
\begin{split}
&\maxi_{G\in\Qsc} \int_0^1-{1\over2}(G(s))^2ds-\int_0^1G(s)d\invbreve H_{\beta,\lambda}(s)\\
&\text{subject to }\quad G(s)\ge0, s\in[0,1).
\end{split}
\end{equation}
We have the following lemma for problems \eqref{opt:dual:G}
and \eqref{opt:dual:G:hat}.

\begin{lemma}\label{lma:hp}
Assume that $X_0\in L^\infty$. Then
for every $\beta\in\Rbb$ and $\lambda>0$, we have that $\invbreve H^\prime_{\beta,\lambda}(1-)>-\infty$. Let
\begin{equation}\label{eq:G^}
G^*_{\beta,\lambda}(s)=\left(-\invbreve H_{\beta,\lambda}^\prime(s)\right)^+,\quad s\in[0,1).
\end{equation}
Then $G^*_{\beta,\lambda}$ is bounded and uniquely solves both problems \eqref{opt:dual:G} and \eqref{opt:dual:G:hat}. Moreover, problems \eqref{opt:dual:G} and \eqref{opt:dual:G:hat} have the same optimal value.
\end{lemma}

\proof Obviously, $H^\prime_{\beta,\lambda}(1-)=2\beta-\lambda Q_\rho(0)-2 Q_0(1-)>-\infty$. Then by Lemma \ref{lma:cv:pfinite}, we have  $\invbreve H^\prime_{\beta,\lambda}(1-)>-\infty$, which implies that $G^*_{\beta,\lambda}$ is bounded. The proof of the rest part is standard in the concave envelope relaxation literature; see, e.g., \citet{Xu2016}, \citet{Wei2018}, and \citet[Lemma 5.5]{WangXia2021}. 
\qed

From Lemma \ref{lma:hp}, we have the following theorem, which provides the solution to the dual optimization problem in closed form.

\begin{theorem}\label{thm:w*}
Assume that $X_0\in L^\infty$. Then for every $\beta\in\Rbb$ and $\lambda>0$,  the dual optimization problem \eqref{opt:dual:w} has a unique solution  $w^*_{\beta,\lambda}$, which satisfies 
$$(w^*_{\beta,\lambda})^\prime(s)=\left(-\invbreve H_{\beta,\lambda}^\prime(s)\right)^+,\quad s\in[0,1),$$
where $H_{\beta,\lambda}$ is defined in \eqref{eq:def:H} and $\invbreve H_{\beta,\lambda}$ is the concave envelope of $H_{\beta,\lambda}$. Moreover, $(w^*_{\beta,\lambda})^\prime$ is bounded and $w^*_{\beta,\lambda}\in\Wicx_{c,2}$.
\end{theorem}

\section{Primal Optimizer}\label{sec:po}

The next theorem provides, for every $(\beta,\lambda)$, the optimal solution to problem \eqref{opt:lambda} in closed form, which is an immediate consequence of combining Theorem \ref{thm:w*}, Lemma \ref{lma:inner}, and Proposition \ref{prop:saddle:L}.

\begin{theorem}\label{thm:Q*}
Assume that $X_0\in L^\infty$. For every $\beta\in\Rbb$ and $\lambda>0$,  let 
$$Q^*_{\beta,\lambda}(s)\trieq \beta+{\left(-\invbreve H_{\beta,\lambda}^\prime(s)\right)^+-\lambda Q_\rho((1-s)-)\over 2},\quad s\in[0,1),$$
where $H_{\beta,\lambda}$ is defined in \eqref{eq:def:H} and $\invbreve H_{\beta,\lambda}$ is the concave envelope of $H_{\beta,\lambda}$. Then $Q^*_{\beta,\lambda}$  is the unique solution to problem \eqref{opt:lambda}. 
\end{theorem}

Recall that, for any $\beta$ and $\lambda$, 
$$H_{\beta,\lambda}(s)=2\beta s-N_\lambda(s),\quad s\in[0,1],$$
where 
$$N_\lambda(s)\trieq\int_0^s(\lambda Q_\rho(1-t)+2Q_0(t))dt,\quad s\in[0,1].$$
Let $\breve  N_\lambda$ denote the convex envelope of $N_\lambda$. Then
$$\invbreve H_{\beta,\lambda}(s)=2\beta s-\breve N_\lambda(s), \quad s\in[0,1]$$
and
$$\invbreve H^\prime_{\beta,\lambda}=2\beta-\breve N^\prime_\lambda.$$
Therefore, 
\begin{equation}\label{eq:Q*N}
Q^*_{\beta,\lambda}(s)= \beta+{\left(\breve N_\lambda^\prime(s)-2\beta\right)^+-\lambda Q_\rho((1-s)-)\over 2},\quad s\in[0,1),
\end{equation}
By Theorem \ref{thm:Q*},
\begin{align*}
 v(\beta,\lambda)=&\int_0^1(Q^*_{\beta,\lambda}(s)-\beta)^2ds+\lambda\int_0^1Q^*_{\beta,\lambda}(s)Q_\rho(1-s)ds\\
 =&{1\over4}\int_0^1\left(\left(\breve N_\lambda^\prime(s)-2\beta\right)^+\right)^2ds+\lambda\beta\ep[\rho]-{\lambda^2\over4}\ep[\rho^2].
\end{align*}
The function $x\mapsto (x^+)^2$ is convex and continuously differentiable. So is $v$ with respect to $\beta$ and 
\begin{equation}\label{eq:hbeta}
{\partial v(\beta,\lambda)\over\partial\beta}=\lambda\ep[\rho]-\int_0^1\left(\breve N_\lambda^\prime(s)-2\beta\right)^+ds\trieq h(\beta,\lambda).
\end{equation}

Finally, we can find the variance-minimal quantile function according to the following two steps.

\begin{description} 
\item[(i)] For any fixed $\lambda>0$, consider the problem 
$$\mini_{\beta\in\Rbb}v(\beta,\lambda).$$
The minimizer $\beta_\lambda$ is determined by $h(\beta_\lambda,\lambda)=0$, where $h$ is given by \eqref{eq:hbeta}. 
Let $Q^*_\lambda\trieq Q^*_{\beta_\lambda,\lambda}$. Then $Q^*_\lambda$  solves problem \eqref{opt:var:lambda}/\eqref{opt:var:lambda:beta}.

\item[(ii)] 
Let 
\begin{align}\label{eq:Xclambda}
\Xc(\lambda)\trieq\int_0^1Q^*_\lambda(s) Q_\rho(1-s)ds,\quad \lambda\in(0,\infty).
\end{align}
By Proposition \ref{prop:existence}, 
$\Xc$ is decreasing on $(0,\infty)$. Moreover, a combination of Theorem \ref{thm:unique}, Proposition \ref{prop:existence} and Remark \ref{rmk: lambdao} implies that, for any $x<Q_0(1)\ep[\rho]$, $\Xc(\lambda^\circ)=x$ for some $\lambda^\circ\in(0,\infty)$.
Therefore, $\Xc$ is continuous on $(0,\infty)$, $\lim_{\lambda\downarrow0}\Xc(\lambda)=Q_0(1)\ep[\rho]$, and 
$\lim_{\lambda\uparrow\infty}\Xc(\lambda)=-\infty$.
The monotonicity and continuity of $\Xc$ makes it easy to search for the desired Lagrange multiplier $\lambda^\circ$ for any given budget level $x<Q_0(1)\ep[\rho]$ by solving the equation $\Xc(\lambda)=x$. 
Once $\lambda^\circ$ has been determined, then $Q^\circ=Q^*_{\lambda^\circ}$ is the desired variance-minimal quantile function in $\Qicx(x,Q_0)$.

\end{description}

\section{Special Case}\label{sec:special}

\subsection{The Classical Case: Constant Benchmark}

Now we consider the classical case when $X_0=z$ a.s. for a constant $z>{x\over\ep[\rho]}$.  In this case,
$$N_\lambda(s)=\lambda \int_0^s Q_\rho(1-s)ds+2zs,\quad s\in[0,1].$$
Obviously, $N_\lambda$ is concave and hence 
$$\breve N^\prime\equiv N_\lambda(1)=\lambda \ep[\rho]+2z.$$
Then
$$h(\beta,\lambda)=
\begin{cases}
2\beta-2z &\text{if }2\beta< \lambda\ep[\rho]+2z,\\
\lambda\ep[\rho] & \text{if }2\beta\ge \lambda\ep[\rho]+2z.
\end{cases}
$$
For any $\lambda>0$, we have $\beta_\lambda=z$ and hence 
$$Q^*_\lambda(s)=z+{\lambda\over2}\ep[\rho]-{\lambda\over2}Q_\rho((1-s)-),\quad s\in[0,1).$$
It is easy to see that
$$\Xc(\lambda)=z\ep[\rho]-{\lambda\over2}\var[\rho].$$
Therefore, ${\lambda^\circ\over2}={z\ep[\rho]-x\over\var[\rho]}$ and hence
$$Q^\circ(s)=z+{z\ep[\rho]-x\over \var[\rho]}\ep[\rho]-{z\ep[\rho]-x\over \var[\rho]}Q_\rho((1-s)-),\quad s\in[0,1).$$

\subsection{Two-Point Distributed Benchmark}

We now consider the special case when $X_0$ is discretely distributed:
\begin{equation}\label{eq:X0:ab}
\Pbb(X_0=a)=p=1-\Pbb(X_0=b),
\end{equation}
where $a\le b$ and $p\in(0,1)$.
In this case,
$$Q_0(s)=a\id_{s<p}+b\id_{s\ge p},\quad s\in[0,1).$$
In Assumption \ref{ass:X0x}, condition $Q_0(1)\ep[\rho]>x$ amounts to $b\ep[\rho]>x$. 

We will use the following notation:
$$A_1\trieq {1\over p}\int_0^p Q_\rho(1-s)ds, \quad A_2\trieq {1\over 1-p}\int_p^1Q_\rho(1-s)ds.$$
Obviously, $A_1>A_2$. 

The next proposition explicitly provides the variance-minimal quantile function. 

\begin{proposition}\label{prop:2point}
For a benchmark $X_0$ that satisfies \eqref{eq:X0:ab}, the variance-minimal quantile function $Q^\circ$ is given as follows.

\begin{description}
\item[(a)] If $x \le \ep[X_0]\ep[\rho]-{b-a\over A_1-A_2}\var[\rho]$, then
$$Q^\circ(s)=\ep[X_0]+{\ep[X_0]\ep[\rho]-x\over\var[\rho]}\ep[\rho]-{\ep[X_0]\ep[\rho]-x\over\var[\rho]}Q_\rho((1-s)-),\quad s\in[0,1).$$

\item[(b)] If $\ep[X_0]\ep[\rho]-{b-a\over A_1-A_2}\var[\rho]< x
\le b\ep[\rho]-{b-a\over A_1}\left((1-p)\ep[\rho^2]+p^2A_1^2-(1-p)^2A_2^2\right)$, then
$$Q^\circ(s)=
\begin{cases}
a+{ap A_1+b(1-p)A_2-x\over\ep[\rho^2]-pA_1^2-(1-p)A_2^2}A_1-{ap A_1+b(1-p)A_2-x\over\ep[\rho^2]-pA_1^2-(1-p)A_2^2}
Q_\rho((1-s)-)&\text{if }s\in[0,p),\\
b+{ap A_1+b(1-p)A_2-x\over\ep[\rho^2]-pA_1^2-(1-p)A_2^2}A_2-{ap A_1+b(1-p)A_2-x\over\ep[\rho^2]-pA_1^2-(1-p)A_2^2}
Q_\rho((1-s)-)&\text{if }s\in[p,1).\\
\end{cases}
$$

\item[(c)] If $b\ep[\rho]-{b-a\over A_1}\left((1-p)\ep[\rho^2]+p^2A_1^2-(1-p)^2A_2^2\right)<x <b\ep\rho]$, then
$$Q^\circ(s)=
\begin{cases}
b-{p(b\ep[\rho]-x)\over(1-p)\ep[\rho^2]+p^2A_1^2-(1-p)^2A_2^2}A_1
-{(1-p)(b\ep[\rho]-x)\over(1-p)\ep[\rho^2]+p^2A_1^2-(1-p)^2A_2^2}Q_\rho((1-s)-)\\
\hskip5cm \text{if }s\in[0,p),\\
b+{(1-p)(b\ep[\rho]-x)\over(1-p)\ep[\rho^2]+p^2A_1^2-(1-p)^2A_2^2}A_2
-{(1-p)(b\ep[\rho]-x)\over(1-p)\ep[\rho^2]+p^2A_1^2-(1-p)^2A_2^2}Q_\rho((1-s)-)\\
\hskip5cm \text{if }s\in[p,1).\\
\end{cases}
$$
\end{description} 
\end{proposition}

\proof See Appendix \ref{app:prop;2point}. \qed

We can reformulate Proposition \ref{prop:2point} as follows.

\begin{proposition}\label{prop:2point:X}
For a benchmark $X_0$ that satisfies \eqref{eq:X0:ab}, the variance-minimal payoff $X^\circ$ is given as follows.

\begin{description}
\item[(a)] If $x \le \ep[X_0]\ep[\rho]-{b-a\over A_1-A_2}\var[\rho]$, then
$$X^\circ=\ep[X_0]+{\ep[X_0]\ep[\rho]-x\over\var[\rho]}\ep[\rho]-{\ep[X_0]\ep[\rho]-x\over\var[\rho]}\rho\text{ a.s.}$$

\item[(b)] If $\ep[X_0]\ep[\rho]-{b-a\over A_1-A_2}\var[\rho]< x
\le b\ep[\rho]-{b-a\over A_1}\left((1-p)\ep[\rho^2]+p^2A_1^2-(1-p)^2A_2^2\right)$, then
$$X^\circ=
\begin{cases}
a+{ap A_1+b(1-p)A_2-x\over\ep[\rho^2]-pA_1^2-(1-p)A_2^2}A_1-{ap A_1+b(1-p)A_2-x\over\ep[\rho^2]-pA_1^2-(1-p)A_2^2}\rho&\text{if }\rho>Q_\rho((1-p)-),\\
b+{ap A_1+b(1-p)A_2-x\over\ep[\rho^2]-pA_1^2-(1-p)A_2^2}A_2-{ap A_1+b(1-p)A_2-x\over\ep[\rho^2]-pA_1^2-(1-p)A_2^2}
\rho&\text{if }\rho\le Q_\rho((1-p)-).\\
\end{cases}
$$

\item[(c)] If $b\ep[\rho]-{b-a\over A_1}\left((1-p)\ep[\rho^2]+p^2A_1^2-(1-p)^2A_2^2\right)<x <b\ep\rho]$, then
$$X^\circ=
\begin{cases}
b-{p(b\ep[\rho]-x)\over(1-p)\ep[\rho^2]+p^2A_1^2-(1-p)^2A_2^2}A_1
-{(1-p)(b\ep[\rho]-x)\over(1-p)\ep[\rho^2]+p^2A_1^2-(1-p)^2A_2^2}\rho\\
\hskip5cm \text{if }\rho> Q_\rho((1-p)-),\\
b+{(1-p)(b\ep[\rho]-x)\over(1-p)\ep[\rho^2]+p^2A_1^2-(1-p)^2A_2^2}A_2
-{(1-p)(b\ep[\rho]-x)\over(1-p)\ep[\rho^2]+p^2A_1^2-(1-p)^2A_2^2}\rho\\
\hskip5cm \text{if }\rho\le Q_\rho((1-p)-).\\
\end{cases}
$$
\end{description} 
\end{proposition}

\begin{example}\label{exm:Xo}
The SDF $\rho$ is log-normal: 
$\log\rho\sim N(\mu,\sigma^2)$ with $\mu=-0.1$ and $\sigma=0.34$. The initial capital $x=1.0$.
\begin{description}
\item[(a)] The benchmark payoff $X_0$ satisfies $\Pbb(X_0=1.1-\delta)=\Pbb(X_0=1.1+\delta)=0.5$ with $\delta=0.10, 0.15,0.20, 0.40$. The variance-minimal payoffs $X^\circ$ v.s. SDF $\rho$ are plotted in Figure \ref{fig:xrho}. 
\item[(b)] The benchmark payoff $X_0$ satisfies $\Pbb(X_0=\alpha-0.30)=\Pbb(X_0=\alpha+0.30)=0.5$ with $\alpha=0.7432, 1.0, 1.15,1.20$. The variance-minimal payoffs $X^\circ$ v.s. SDF $\rho$ are plotted in Figure \ref{fig:xrhoalpha}. Here, $\alpha=0.7432$ is the solution of $Q_0(1)\ep[\rho]=(\alpha+0.30)\ep[\rho]=x$; in this case, $X^\circ$ is a constant. 
\end{description}
\end{example}

\begin{figure}[!ht]
\centering
\includegraphics[scale=0.3]{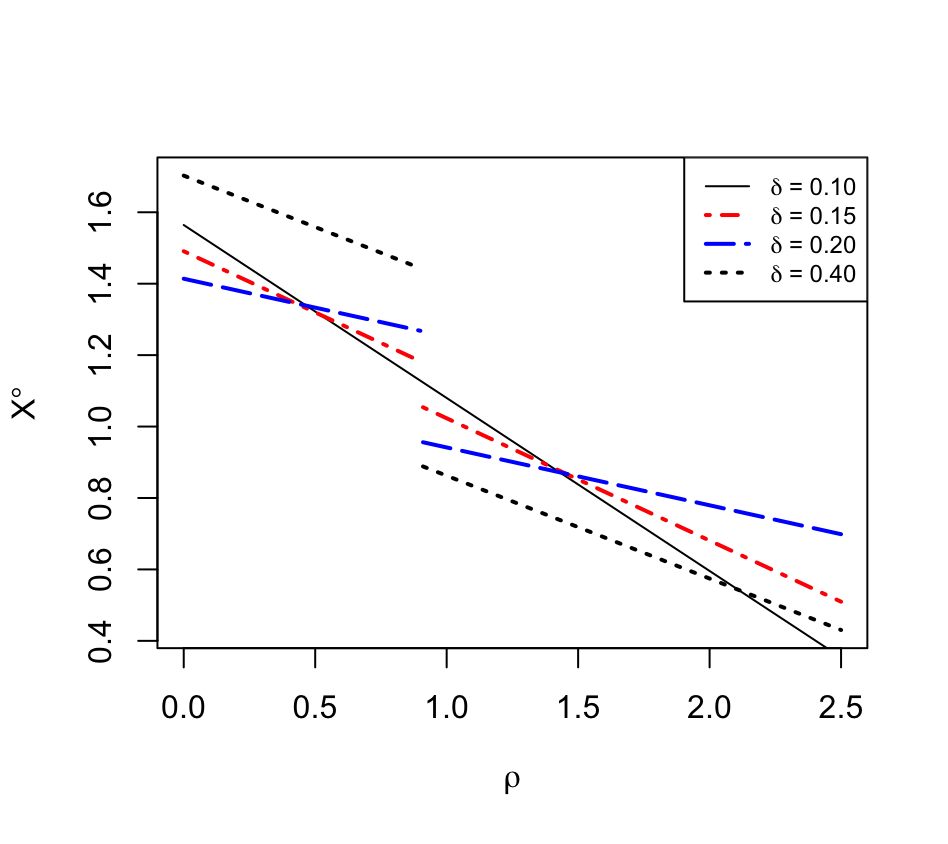}
\caption{\small
Variance-minimal payoffs v.s. SDF
}\label{fig:xrho}
\end{figure}

\begin{figure}[!ht]
\centering
\includegraphics[scale=0.3]{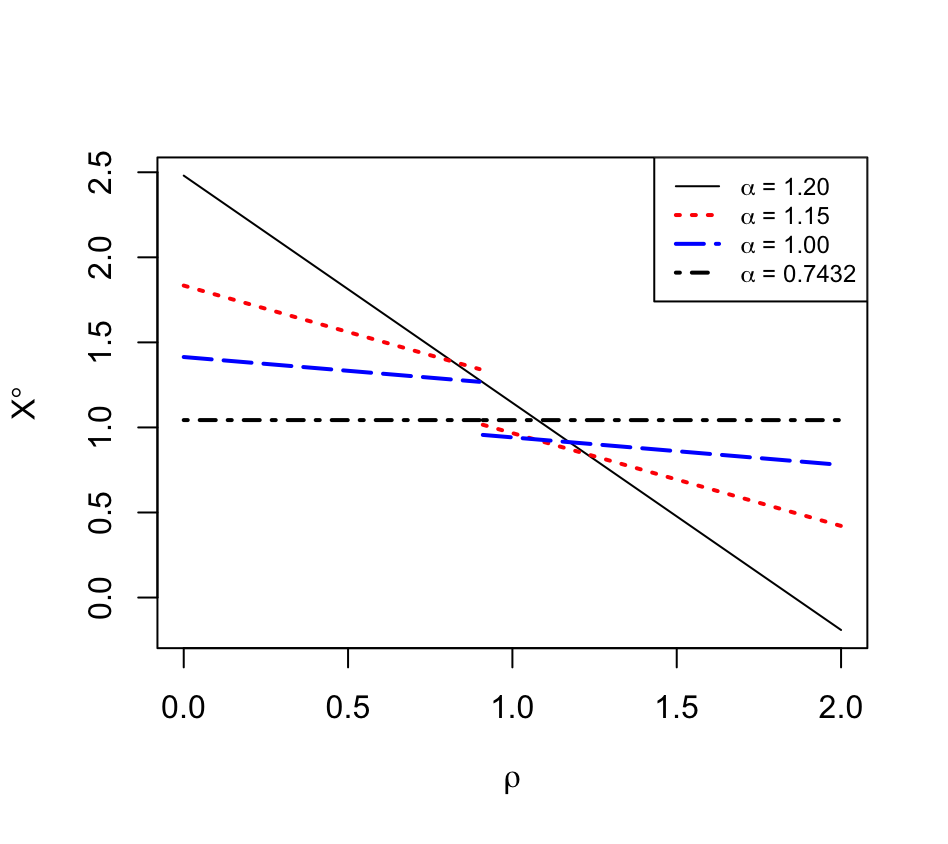}
\caption{\small
Variance-minimal payoffs v.s. SDF
}\label{fig:xrhoalpha}
\end{figure}

\section{Discussions}\label{sec:disc}

\subsection{Beating-Performance-Variance Efficient Payoffs}\label{sec:PV}

Consider a given benchmark payoff $X_0\in L^\infty$. For any payoff $X\in L^2$, its performance of beating $X_0$ is\footnote{
The benchmark-beating performance is similar to but different from the SSD-based risk measure of  \citet{Fabian2011}, which is also called benchmark-based expected shortfall as a special case of adjusted expected shortfall in \citet{BMW2022}. 
Given a benchmark payoff $X_0$, the SSD-based risk measure or the benchmark-adjusted expected shortfall of a payoff $X$ is $\Rc(X)\trieq \inf\{m\in\Rbb\mid X-m\icv X_0\}$.
}
$$\psi(X)\trieq\sup\{m\in\Rbb\mid X-m\icx X_0\}\quad\text{with }\sup\emptyset=-\infty.$$ 
Obviously, the beating performance $\psi$ satisfies the following properties.
\begin{itemize} 
\item Monotonicity:  $\psi(X)\ge\psi(Y)$ if $X\ge Y$ a.s.
\item Translation Invariance: $\psi(X+c)=\psi(X)+c$ for all $X\in L^2$ and $c\in\Rbb$.
\item Law Invariance: $\psi(X)=\psi(Y)$ if $X$ and $Y$ are identically distributed. 
\end{itemize}
Moreover,  $\psi$ has the following representation:
\begin{align*}
\psi(X)=\inf_{t\in(0,1)}\left({1\over t}\int_{1-t}^1 (Q_X(s)-Q_0(s))ds\right),\quad X\in L^2.
\end{align*}
Clearly, $\psi(X)\in[-\infty,\infty)$ for all $X\in L^2$ and
$$\psi(c)=c-Q_0(1)\quad\text{for all } c\in\Rbb.$$
It is easy to see that, for all $X\in L^2$ and $z\in\Rbb$,
$$\psi(X)\ge z\Longleftrightarrow X\icx X_0+z.$$
In particular, $\psi(X)\ge0$ if and only if $X\icx X_0$. 

For any budget level $x\in\Rbb$, let
$$\Xsc(x)\trieq\{X\in L^2\mid \ep[\rho X]\le x\}.$$

\begin{definition}
A payoff $X\in \Xsc(x)$
is called beating-performance-variance (BPV) efficient in $\Xsc(x)$ if there is no $Y\in\Xsc(x)$
such that 
\[
\psi(Y)\geq\psi(X)\quad\text{and}\quad\var[Y]\leq\var[X]
\]
 with at least one inequality holding strictly. 
\end{definition}

\begin{remark} If $X_0=0$, then $\psi(X)=\ep[X]$ for all $X\in L^2$. In this case, the BPV efficiency reduces to the mean-variance efficiency. 
\end{remark}

\begin{proposition}\label{prop:pv}
Assume $X_0\in L^\infty$. Then
$X^*$ is  BPV efficient in $\Xsc(x)$ if and only if $X^*$ is variance-minimal in $\Xicx(x,X_0+z)$  for some $z\ge {x\over\ep[\rho]}-Q_0(1)$.
\end{proposition}

\proof Assume that $X^*$ is  BPV efficient in $\Xsc(x)$. Let $z=\psi(X^*)$. By ${x\over\ep[\rho]}\in\Xsc(x)$ and $\var\left[{x\over\ep[\rho]}\right]=0\le\var(X^*)$ and by the BPV efficiency of $X^*$, we have 
$$z=\psi(X^*)\ge\psi\left({x\over\ep[\rho]}\right)={x\over\ep[\rho]}-Q_0(1).$$ 
Now we show that $X^*$ is variance-minimal in $\Xicx(x,X_0+z)$.  
Suppose, on the contrary, that there exists some $Y\in\Xicx(x,X_0+z)$ such that $\var(Y)<\var(X^*)$. 
Then $\psi(Y)\ge z=\psi(X^*)$ and $Y\in\Xsc(x)$, which contradicts the BPV efficiency of $X^*$. 

Conversely, assume that $X^*$ is variance-minimal in $\Xicx(x,X_0+z)$  for some $z\ge {x\over\ep[\rho]}-Q_0(1)$. We are going to show 
$X^*$ is BPV efficient in $\Xsc(x)$. The discussion is divided into two cases.
\begin{description}
\item[(a)] Assume $z={x\over\ep[\rho]}-Q_0(1)$, i.e., 
$(Q_0(1)+z)\ep[\rho]=x$. In this case, Proposition \ref{prop:trivial} implies that $\var[X^*]=0$ and hence $X^*=c$ a.s. for some $c\in\Xicx(x,X_0+z)$. Then, 
$c\ep[\rho]\le x$ and $c\icx X_0+z$. Therefore, $c\ge Q_0(1)+z$ and hence $c\ep[\rho]\ge (Q_0(1)+z)\ep[\rho]=x$,
implying $c\ep[\rho]=x$. Suppose on the contrary that $X^*$ is not BPV efficient in $\Xsc(x)$, i.e.,  there exists some $Y\in\Xsc(x)$ such that 
$$\var(Y)\le\var(X^*)=0\quad\text{and}\quad \psi(Y)\ge\psi(X^*)$$
with at least one inequality holds strictly.  Then $Y=c_0$ a.s. for some $c_0\in\Rbb$ and $\psi(c_0)>\psi(X^*)=\psi(c)$.
Therefore, $c_0>c$ and hence $c_0\ep[\rho]>c\ep[\rho]=x$, which is impossible since $Y=c_0\in\Xsc(x)$. Thus, $X^*$ is BPV efficient in $\Xsc(x)$.

\item[(b)] Assume $z>{x\over\ep[\rho]}-Q_0(1)$, i.e., 
$(Q_0(1)+z)\ep[\rho]>x$. 
Let $Y\in\Xsc(x)$ and $\psi(Y)\ge\psi(X^*)$. 
Then 
$\psi(Y)\ge\psi(X^*)\ge z$ 
and hence $Y\in\Xicx(x,X_0+z)$. By the uniqueness of the variance-minimal payoff (Theorem \ref{thm:unique}), $\var(Y)>\var(X^*)$ unless $Y=X^*$ a.s.  Thus, $X^*$ is BPV efficient in $\Xsc(x)$.
\qed
\end{description}

\begin{proposition}\label{prop:psiz}
Assume $X_0\in L^\infty$.  Let $z\ge {x\over\ep[\rho]}-Q_0(1)$. If 
$X^*$ is variance-minimal in $\Xicx(x,X_0+z)$, then $\psi(X^*)=z$.
\end{proposition}

\proof Consider the case when $z={x\over\ep[\rho]}-Q_0(1)$.  By the proof of Proposition \ref{prop:pv}, we know $X^*={x\over \ep[\rho]}$ a.s. and hence $\psi(X^*)={x\over \ep[\rho]}-Q_0(1)=z$.

Now we consider the case when $z>{x\over\ep[\rho]}-Q_0(1)$, i.e., 
$(Q_0(1)+z)\ep[\rho]>x$. In this case, $\var[X^*]>0$ by Theorem \ref{thm:unique}. 
Obviously, $\psi(X^*)\ge z$ by $X^*\icx X_0+z$. It is left to show $\psi(X^*)\le z$. 
Suppose on the contrary that $\psi(X^*)> z$. Then there exists some $\alpha>0$ such that
$X^*-\alpha\icx X_0+z$. Therefore, 
$$X^\varepsilon\trieq(1-\varepsilon)(X^*-\alpha)+\varepsilon(Q_0(1)+z)\in\Xicx(x,X_0+z)$$
for all sufficiently small $\varepsilon>0$. But $\var[X^\varepsilon]=(1-\varepsilon)^2\var[X^*]<\var[X^*]$, contradicting the variance-minimality of $X^*$. Therefore, $\psi(X^*)=z$. \qed

Based on Propositions \ref{prop:pv} and \ref{prop:psiz}, the \textit{beating-performance-standard-deviation (BPSD) efficient frontier} is 
$$\left\{\Big(\sqrt{\var(X^*(z)},z\Big)\,\left|\, z\ge  {x\over\ep[\rho]}-Q_0(1)\right.\right\},$$
where, for every $z$,  $X^*(z)$ is the variance-minimal payoff in $\Xicx(x,X_0+z)$.

\begin{example}
The SDF $\rho$ and the initial capital $x$ are same to Example \ref{exm:Xo}. 
The benchmark payoff $X_0$ satisfies $\Pbb(X_0=-\delta)=\Pbb(X_0=\delta)=0.5$ with $\delta=0.0,0.2, 0.5$.
The BPSD efficient frontiers are plotted in Figure \ref{fig:frontier}. In particular,  $\delta=0.0$ refers to the classical mean-standard-deviation efficient frontier, which is a straight line. 
\end{example}

\begin{figure}[!ht]
\centering
\includegraphics[scale=0.3]{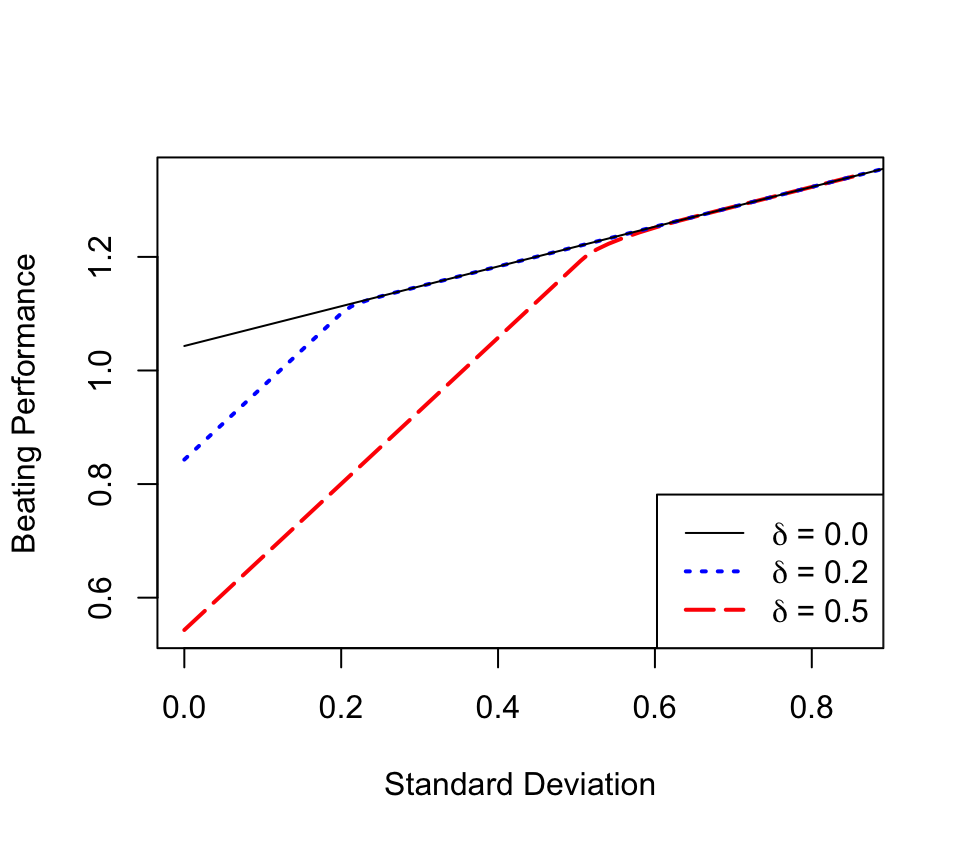}
\caption{\small
BPSD frontiers
}\label{fig:frontier}
\end{figure}

\subsection{Multi-Benchmark Beating}\label{sec:multi} 

In the previous discussion, we considered the problem with beating only one benchmark.
Now we consider the problem with beating multiple benchmarks. It turns out that the multi-benchmark case can be reduced to the single-benchmark case. Actually, in the quantile formulation, we can write beating 
multiple benchmarks as 
\begin{equation*}
\int_t^1 Q(s)ds\geq \int_t^1 Q_j(s)ds\trieq f_j(t) \quad\mbox{for all } t\in[0,1] \mbox{ and } 1\le j\le k,\\
\end{equation*}
where $Q_j\in\Qsc$ are bounded for all $1\le j\le k$.
This system of $k$ constraints is equivalent to
$$\int_t^1 Q(s)ds\geq \max\left\{\int_t^1Q_1(s)ds,\ldots,\int_t^1 Q_k(s)ds\right\}\trieq g(t) \quad\mbox{for all } t\in[0,1].$$

Obviously, $f_j\le g$, and $f_j(1)=g(1)=0$ for all $j$. There exists a sequence $\{t_n\}_{n\ge1}\subset(0,1)$ such that $\lim_{n\to\infty}t_n=1$ and
$$\lim_{n\to\infty}{g(t_n)-g(1)\over t_n-1}=\liminf_{t\to 1}{g(t)-g(1)\over t-1}.$$
There exists some $j_0$ such that $g(t_n)=f_{j_0}(t_n)$ for infinitely many $n$. In this case,
\begin{align*}
\lim_{n\to\infty}{g(t_n)-g(1)\over t_n-1}\ge\liminf_{n\to\infty}{f_{j_0}(t_n)-f_{j_0}(1)\over t_n-1}=f^\prime_{j_0}(1-)=-Q_{j_0}(1-).
\end{align*}
Then $$\liminf_{t\to 1}{g(t)-g(1)\over t-1}\ge-Q_{j_0}(1-)>-\infty.$$ 
Similarly, $\limsup_{t\to 0}{g(t)-g(0)\over t}<\infty$.

Let $\invbreve g$ be the concave envelope of $g$. By Lemma \ref{lma:cv:pfinite}, 
$\infty>\invbreve g^\prime(0)\ge \invbreve g^\prime(1-)>-\infty$. Therefore, $-\invbreve g^\prime$ is bounded and increasing. Let 
$$\bar Q(t)=-\invbreve g^\prime(t),\quad t\in[0,1).$$ 
Then $\bar Q$ is the quantile function of some $\bar X\in L^\infty$ and $\invbreve g(t)=\int_t^1\bar Q(s)ds$ for all $t\in[0,1]$.

For any $Q\in\Qsc$, let $f(t)=\int_t^1Q(s)ds$, $t\in[0,1]$. Obviously, $f$ is a concave function. By the definition of concave envelope, we know that
$$f\ge g\Longleftrightarrow f\ge \invbreve g.$$
On the other hand,
$$Q\icx Q_j,\; \forall j \Longleftrightarrow f\ge g$$
and
$$Q\icx \bar Q \Longleftrightarrow f\ge \invbreve g.$$
Therefore,
$$Q\icx Q_j,\; \forall j \Longleftrightarrow Q\icx\bar Q.$$
This shows that the multi-benchmark case can be reduced to the single-benchmark case.

\subsection{Mean-Variance Efficient Payoff}

Now we consider the following problem of mean-variance portfolio selection with the increasing convex order constraint:
\begin{equation}\label{opt:mv:Q:z}
\begin{split}
&\mini_{Q\in\Qicx(x,Q_0)}\quad \var[Q]\\
&\text{ subject to }\quad \ep[Q]\ge z,
\end{split}
\end{equation}  
where $z\in\Rbb$. 
Such a mean-variance problem can be covered by the analysis in the previous sections. 

Actually, in the case when $z\le \ep[Q^\circ]$ for some variance-minimal quantile function $Q^\circ\in\Qicx(x,Q_0)$,
 $Q^\circ$ solves problem \eqref{opt:mv:Q:z}. 
 
 Now assume that $z> \ep[Q^\circ]$ for all variance-minimal quantile function $Q^\circ\in\Qicx(x,Q_0)$. In this case, 
by the convexity of variance and by the method of Lagrangian multiplier,  $Q^*$ solves problem \eqref{opt:mv:Q:z} if and only if, for some $\gamma^*\ge0$, $Q^*$ solves the following problem
\begin{equation}\label{opt:gamma}
\mini_{Q\in\Qicx(x,Q_0)}\quad \var[Q]-2\gamma^*\ep[Q]
\end{equation}
and 
$$\gamma^*(\ep[Q^*]-z)=0.$$
It is easy to see that $\gamma^*=0$ is impossible. Then $\gamma^*>0$ and hence $\ep[Q^*]=z$. In this case, problem \eqref{opt:mv:Q:z} is equivalent to
\begin{equation}\label{opt:mv:Q:z=}
\begin{split}
&\mini_{Q\in\Qicx(x,Q_0)}\quad \var[Q]=\int_0^1Q^2(s)ds-z^2\\
&\text{ subject to }\quad \ep[Q]= z,
\end{split}
\end{equation}  
which is further equivalent to
\begin{equation}\label{opt:mv:Q:z=:gamma}
\begin{split}
\mini_{Q\in\Qicx(x,Q_0)}\quad \int_0^1Q^2(s)ds-2\gamma^*\int_0^1Q(s)ds.
\end{split}
\end{equation}  
By the method of Lagrangian multiplier once again, problem \eqref{opt:mv:Q:z=:gamma} can be transformed to, for some $\lambda\ge0$,
\begin{equation}\label{opt:mv:Q:z=:gamma:lambda}
\begin{split}
&\mini_{Q\in\Qsc}\quad \int_0^1Q^2(s)ds-2\gamma^*\int_0^1Q(s)ds+\lambda\int_0^1Q(s)Q_\rho(1-s)ds\\
&\text{subject to }\quad Q\icx Q_0.
\end{split}
\end{equation} 
It is just problem \eqref{opt:lambda} with $\beta=\gamma^*$.

As Remark \ref{rmk:z} indicates,  we can also read $\ep[Q]\ge z$ as $Q\icx z$ and hence the mean-variance problem \eqref{opt:mv:Q:z} is equivalent to a variance minimizing problem with two benchmarks: $X_0$ and $z$. 
Then it can be reduced to a variance minimizing problem with one benchmark, as discussed in Section \ref{sec:multi}.

\newpage
\begin{appendices}

\section{Some Technical Results}

The results in this section may be unoriginal but are provided here for convenience nevertheless.

\subsection{Saddle Point}

The following lemma is essentially an abstract version of the proof of \citet[Proposition 5.3]{WangXia2021}. 
It is provided here for the convenience of its applications, not only in the proof of Proposition \ref{prop:saddle:L1} here but also elsewhere.

\begin{lemma}\label{lma:fgh}
Let $X$ be a nonempty set and $Y_1$ be a nonempty subset of a linear space.  Let $Y$ be the cone generated by $Y_1$, that is, $Y=\{\mu y\mid \mu\ge0,\; y\in Y_1\}$.
Consider a function $f: X\times Y\to\Rbb$. Let functions $g:X\times Y_1\times [0,\infty)\to\Rbb$ and $h: X\times[0,\infty)\to\Rbb$ be given by
\begin{align*}
&g(x,y,\mu)=f(x,\mu y),\quad x\in X, y\in Y_1, \mu\ge0,\\
&h(x,\mu)=\sup_{y\in Y_1}g(x,y,\mu)=\sup_{y\in Y_1}f(x,\mu y),\quad x\in X,\mu\ge0.
\end{align*}
Assume that both of the following two conditions are satisfied.
\begin{description}
\item[(a)] For every $\mu\ge0$, there exists some $y_1(\mu)\in Y_1$ such that
\begin{equation*}
\inf_{x\in X}g(x,y_1(\mu),\mu)=\sup_{y\in Y_1}\inf_{x\in X}g(x,y,\mu)=\inf_{x\in X}\sup_{y\in Y_1}g(x,y,\mu).
\end{equation*}
\item[(b)] $(x^*,\mu^*)$ is a saddle point of $h$, that is,
$$h(x,\mu^*)\ge h(x^*,\mu^*)\ge h(x^*,\mu),\quad \forall x\in X,\mu\ge0.$$
\end{description}
Let $y^*=\mu^*y_1(\mu^*)$. Then $(x^*,y^*)$ is a saddle point of $f$, that is,
$$f(x,y^*)\ge f(x^*,y^*)\ge f(x^*,y),\quad \forall x\in X, y\in Y.$$
\end{lemma}

\proof  We will frequently use conditions (a) and (b). Firstly, we have
\begin{align*}
&h(x^*,\mu^*)=\sup_{\mu\ge0}\inf_{x\in X}h(x,\mu)=\sup_{\mu\ge0}\inf_{x\in X}\sup_{y\in Y_1}g(x,y,\mu)\\
=&\sup_{\mu\ge0}\sup_{y\in Y_1}\inf_{x\in X}g(x,y,\mu)=\sup_{\mu\ge0}\sup_{y\in Y_1}\inf_{x\in X}f(x,\mu y)=\sup_{y\in Y}\inf_{x\in X}f(x,y).
\end{align*}
 Secondly,  we have
 \begin{align*}
 h(x^*,\mu^*)=\inf_{x\in X}h(x,\mu^*)=\inf_{x\in X}\sup_{y\in Y_1}g(x,y,\mu^*)=\inf_{x\in X}g(x,y_1(\mu^*),\mu^*)=\inf_{x\in X}f(x,y^*).
 \end{align*}
 Thirdly,  we have
\begin{align*}
h(x^*,\mu^*)=\inf_{x\in X}\sup_{\mu\ge0}h(x,\mu)=\inf_{x\in X}\sup_{\mu\ge0}\sup_{y\in Y_1}f(x,\mu y)=\inf_{x\in X}\sup_{y\in Y}f(x,y).
\end{align*}
Fourthly,  we have
\begin{align*}
h(x^*,\mu^*)=\sup_{\mu\ge0}h(x^*,\mu)=\sup_{\mu\ge0}\sup_{y\in Y_1}f(x^*,\mu y)=\sup_{y\in Y}f(x^*,y).
\end{align*}
From the above discussion, we have
$$\inf_{x\in X}f(x,y^*)=\sup_{y\in Y}\inf_{x\in X}f(x,y)=\inf_{x\in X}\sup_{y\in Y}f(x,y)=\sup_{y\in Y}f(x^*,y).$$
Therefore, $(x^*,y^*)$ is a saddle point of $f$. \qed

\subsection{Concave Envelope}

For a proof of the following lemma, see, e.g., \citet[Appendix A.3]{WangXia2021}. 
\begin{lemma}\label{lam:cv:envelope}
Assume that $H: [0,1]\to\Rbb$ is upper semi-continuous. Let $\invbreve H$ be the concave envelope of $H$, i.e.,
$$\invbreve H(s)\trieq\inf\{G(s)\,|\, G \mbox{ is concave
and }G\ge H\mbox{ on }[0,1]\},\quad s\in[0,1].$$ 
We have the following assertions:
\begin{description}
\item[(a)] $\invbreve H(0)=H(0)$ and $\invbreve H(1)=H(1)$;
\item[(b)] $\invbreve H$ is continuous on $[0,1]$;
\item[(c)] $\invbreve H$ is affine on $[\invbreve H>H]$.
\end{description}
\end{lemma}

\begin{lemma}\label{lma:cv:pfinite}
Assume that $H: [0,1]\to\Rbb$ is upper semi-continuous. Let $\invbreve H$ be the concave envelope of $H$. Then we have the following two assertions.
\begin{description}
\item[(a)] If $\liminf_{t\to1}{H(t)-H(1)\over t-1}>-\infty$, then $\invbreve H^\prime(1-)>-\infty$. 
\item[(b)] If $\limsup_{t\to0}{H(t)-H(0)\over t}<\infty$,
then $\invbreve H^\prime(0)<\infty$.
\end{description}
\end{lemma}

\proof We only prove assertion (b), since the proof of assertion (a) is similar. 
Assume that there exists some $\varepsilon\in(0,1)$ such that 
\begin{equation}\label{eq:hhe}
(\varepsilon,1)\subseteq [\invbreve H>H].
\end{equation}
In this case,
$\invbreve H^\prime$ is constant and finite on $(\varepsilon,1)$. Therefore, $\invbreve H^\prime(1-)>-\infty$.  Assume that \eqref{eq:hhe} holds for no $\varepsilon\in(0,1)$. In this case, there exits a sequence $\{t_n\}_{n\ge1}\subset(0,1)$ such that $\lim_{n\to\infty}t_n=1$ and $\invbreve H(t_n)=H(t_n)$ for all $n\ge1$. Then by $\invbreve H(1)=H(1)$, we have that 
\begin{align*}
\invbreve H^\prime(1-)=\lim_{n\to\infty}{\invbreve H(t_n)-\invbreve H(1)\over t_n-1}=\lim_{n\to\infty}{H(t_n)-H(1)\over t_n-1}\ge\liminf_{t\to1}{H(t)-H(1)\over t-1}>-\infty.
\end{align*}
 \qed

\section{Some Proofs}

\subsection{Proof of Theorem \ref{thm:emp}}\label{app:proof:thm:emp}

 For any $\varepsilon\in\left(0,{1\over2}\right)$, let 
$$a_\varepsilon=\int_0^\varepsilon Q_\rho(s)ds=\int_{1-\varepsilon}^1Q_\rho(1-s)ds,\quad \varepsilon\in\left(0,{1\over2}\right).$$
For any $\varepsilon\in\left(0,{1\over2}\right)$ and $n\ge1$, let $Q^\varepsilon_n$ be defined by
$$Q^\varepsilon_n(s)=Q_0(s)+{n\over a_\varepsilon}\id_{s>1-\varepsilon}-{n\over a_\varepsilon}\id_{s<\varepsilon},\quad s\in(0,1),$$
Obviously, $Q^\varepsilon_n\in\Qsc$ and $Q^\varepsilon_n\icx Q_0$ for all $\varepsilon\in\left(0,{1\over2}\right)$ and $n\ge1$.
Moreover, since $Q_\rho(1-)>Q_\rho(0)$, there exists some $\varepsilon_0\in(0,{1\over2})$ and $\alpha>0$ such that
$${1\over a_{\varepsilon_0}}\int_0^{\varepsilon_0} Q_\rho(1-s)ds
={\int_0^{\varepsilon_0} Q_\rho(1-s)ds\over \int_0^{\varepsilon_0} Q_\rho(s)ds}\ge 1+\alpha.$$
Therefore, for any $n\ge1$,
\begin{align*}
&\int_0^1Q^{\varepsilon_0}_n(s)Q_\rho(1-s)ds\\
=&\int_0^1Q_0(s)Q_\rho(1-s)ds+{n\over a_{\varepsilon_0}}\int_{1-\varepsilon_0}^1Q_\rho(1-s)ds- {n\over a_{\varepsilon_0}}\int_0^{\varepsilon_0} Q_\rho(1-s)ds\\
\le&\int_0^1Q_0(s)Q_\rho(1-s)ds+n-n (1+\alpha)\\
=&\int_0^1Q_0(s)Q_\rho(1-s)ds-\alpha n.
\end{align*}
As a consequence, $\inf_{Q\in\Qsc}\int_0^1Q(s)Q_\rho(1-s)ds=-\infty$.
\qed

 \subsection{Proof of Theorem \ref{thm:unique}}\label{app:thm:unique}
 
It is well known that $\var[X]=\min_{\beta\in\Rbb}\ep[(X-\beta)^2]$ for every $X\in L^2$. Then 
$\var[Q]=\min_{\beta\in\Rbb}\int_0^1(Q(s)-\beta)^2ds$ for every $Q\in \Qsc$. Therefore, problem \eqref{opt:vm:Q} can be rewritten as
\begin{equation}\label{opt:vm:Q:beta0}
\begin{split}
&\mini_{Q\in\Qsc}\quad \min_{\beta\in\Rbb}\int_0^1(Q(s)-\beta)^2ds\\
& \text{subject to }\quad  \int_0^1Q(s)Q_\rho(1-s)ds\le x,\; Q\icx Q_0.
\end{split}
\end{equation}

We now consider,
for any fixed $\beta\in\Rbb$, the following problem
\begin{equation}\label{opt:vm:Q:beta1}
\begin{split}
&\mini_{Q\in\Qsc}\quad \int_0^1(Q(s)-\beta)^2ds\\
& \text{subject to }\quad  \int_0^1Q(s)Q_\rho(1-s)ds\le x,\; Q\icx Q_0.
\end{split}
\end{equation}
Let $v_1(\beta)$ denote the optimal value of problem \eqref{opt:vm:Q:beta1}. Obviously, $v^\circ(x)=\inf_{\beta\in\Rbb}v_1(\beta)$.

\begin{lemma}\label{lma:beta} Under Assumption \ref{ass:X0x}, we have the following assertions.
\begin{description}
\item[(a)] For every $\beta\in\Rbb$, there exists a unique optimal solution to problem \eqref{opt:vm:Q:beta1}. 
\item[(b)] $v_1$ is continuous and convex on $\Rbb$ with $\lim_{|\beta|\to\infty}v_1(\beta)=\infty$. 
\item[(c)] There exists some $\beta^*\in\Rbb$ such that $v_1(\beta^*)=\inf_{\beta\in\Rbb}v_1(\beta)$. 
\end{description}
\end{lemma} 

\proof  
We have known that $\Qicx(x,Q_0)$ is a nonempty, convex, and closed subset of Hilbert space $L^2([0,1))$. Problem \eqref{opt:vm:Q:beta1} is nothing but to find an element $Q\in\Qicx(x,Q_0)$ which is closest to $\beta$. Such an element always exists and is unique; see, e.g., \citet[Section 3.12]{L1969}. Therefore, assertion (a) is proved.

Assertion (c) is an easy implication of assertion (b). We now prove assertion (b). 

We firstly prove the convexity. It suffices to show $v_1\left({1\over2}\beta_1+{1\over2}\beta_2\right)\le {1\over2}v_1(\beta_1)+{1\over2}v_1(\beta_2)$ for all $\beta_1,\beta_2\in\Rbb$.  We use $\|\cdot\|$ to denote the norm on Hilbert space $L^2([0,1))$. Let $\beta_1,\beta_2\in\Rbb$. By assertion (a), there exist some $Q_i\in\Qicx(x,Q_0)$ such that $v_1(\beta_i)=\|Q_i-\beta_i\|^2$ for $i\in\{1,2\}$. Obviously,  ${1\over2}(Q_1+Q_2)\in\Qicx(x,Q_0)$ and
\begin{align*}
v_1\left({1\over2}\beta_1+{1\over2}\beta_2\right)\le&\left\|{1\over2}(Q_1+Q_2)-{1\over2}(\beta_1+\beta_2)\right\|^2
=\left\|{1\over2}(Q_1-\beta_1)+{1\over2}(Q_2-\beta_2)\right\|^2\\
\le&{1\over2}\|Q_1-\beta_1\|^2+{1\over2}\|Q_2-\beta_2\|^2={1\over2}v_1(\beta_1)+{1\over2}v_1(\beta_2).
\end{align*}
Therefore, $v_1$ is convex on $\Rbb$. 

Secondly, because $v_1$ is convex and finite on $\Rbb$, $v_1$ is continuous on $\Rbb$.

Finally, we show $\lim_{|\beta|\to\infty}v_1(\beta)=\infty$. To this end, for any $\beta\in\Rbb$, let
\begin{equation}\label{opt:v0}
\begin{split}
v_0(\beta)=&\inf_{Q\in\Qsc} \int_0^1(Q(s)-\beta)^2ds\\
& \text{subject to }  \int_0^1Q(s)Q_\rho(1-s)ds\le x.
\end{split}
\end{equation}
It is a classical problem of mean-variance portfolio selection without constraint. The optimal solution to problem \eqref{opt:v0} is obviously given by
	\[
	Q(s)=\beta-\frac{\lambda Q_{\rho}((1-s)-)}{2},\quad s\in[0,1),
	\]
where $\lambda={2(\beta\ep[\rho]-x)\over\ep[\rho^2]}$. Then
$v_0(\beta)=\frac{(\beta\ep[\rho]-x)^{2}}{\ep[\rho^2]}$ and hence
$$\lim_{|\beta|\to\infty}v_1(\beta)\ge \lim_{|\beta|\to\infty}v_0(\beta)=\infty.$$
Therefore, $\lim_{|\beta|\to\infty}v_1(\beta)=\infty$.\qed

\paragraph{Proof of Theorem \ref{thm:unique}}
 Lemma \ref{lma:beta} guarantees the existence of a variance-minimal quantile function in $\Qicx(x,Q_0)$.  
 
 We now show $v^\circ(x)>0$. Otherwise, $v^\circ(x)=0$ and hence there exists some constant $c\in\Qicx(x,Q_0)$. Then $c\ge Q_0(1)$ and hence $x\ge \int_0^1 c Q_\rho(1-s)ds\ge Q_0(1)\ep[\rho]$, contradicting the assumption that $Q_0(1)\ep[\rho]>x$. Therefore, $v^\circ(x)>0$.

It is left to show the uniqueness. Let $Q_1,Q_2\in \Qicx(x,Q_0)$ be variance-minimal in $\Qicx(x,Q_0)$. We need to show $Q_1=Q_2$.

We first show $Q_1-\ep[Q_1]=Q_2-\ep[Q_2]$. Actually, $\var[Q_1]=\var[Q_2]=v^\circ(x)$. Let $\bar Q={1\over2}(Q_1+Q_2)$. Then $\bar Q\in\Qicx(x,Q_0)$. Moreover, 
\begin{align*}
v^\circ(x)\le&\var[\bar Q]=\ep[(\bar Q-\ep[\bar Q])^2]\\
=&\ep\left[\left({1\over 2}(Q_1-\ep[Q_1])+{1\over2}(Q_2-\ep[Q_2])\right)^2\right]\\
\le&\ep\left[{1\over 2}(Q_1-\ep[Q_1])^2+{1\over2}(Q_2-\ep[Q_2])^2\right]\\
=&{1\over2}\var[Q_1]+{1\over2}\var[Q_2]=v^\circ(x)
 \end{align*}
and hence
 $$\ep\left[\left({1\over 2}(Q_1-\ep[Q_1])+{1\over2}(Q_2-\ep[Q_2])\right)^2\right]
=\ep\left[{1\over 2}(Q_1-\ep[Q_1])^2+{1\over2}(Q_2-\ep[Q_2])^2\right].$$
By the strict convexity of function $x\mapsto x^2$, we have $Q_1-\ep[Q_1]=Q_2-\ep[Q_2]$. 

To show $Q_1=Q_2$, we needs to show $\ep[Q_1]=\ep[Q_2]$. To this end, it suffices to show 
$\int_0^1Q_1(s)Q_\rho(1-s)ds=\int_0^1Q_2(s)Q_\rho(1-s)ds=x$. Without loss of generality, suppose on the contrary that $\int_0^1Q_1(s)Q_\rho(1-s)ds<x$. For $\varepsilon\in(0,1)$, let $Q^\varepsilon$ be given by
$$Q^\varepsilon=(1-\varepsilon)Q_1+\varepsilon Q_0(1).$$
Then for all sufficiently small $\varepsilon>0$, $Q^\varepsilon\in\Qicx(x, Q_0)$. By $\var[Q_1]=v^\circ(x)>0$, we have
$\var[Q^\varepsilon]=(1-\varepsilon)^2\var[Q_1]<\var[Q_1]$. It is impossible since $Q_1$ is variance-minimal. Therefore,
$\int_0^1Q_1(s)Q_\rho(1-s)ds=\int_0^1Q_2(s)Q_\rho(1-s)ds=x$. \qed

\subsection{Proof of Proposition \ref{prop:saddle:L}}\label{app:proof:saddle:L}

Let
\begin{align*}
\Msc\trieq\left\{m:[0,1]\to[0,\infty)\mid  m \mbox{ is
 increasing and right-continuous}\right\}.
\end{align*}
We can identify $\Msc$ as the set of all finite measures on the measurable space $\left([0,1], \Bc_{[0,1]}\right)$.
For each $m\in\Msc$, the measure of $\{0\}$ is $m(0)$ and the measure of $\{1\}$ is $m(1)-m(1-)$.
Let
\begin{align*}
\Msc_1\trieq\left\{m\in\Msc\mid  m(1)=1 \right\}.
\end{align*}
We can identify $\Msc_1$ as the set of all probability measures on $\left([0,1], \Bc_{[0,1]}\right)$.
Consider the weak topology (the topology of weak convergence of measures) on $\Msc_1$, which is induced by all real-valued continuous functions on $[0,1]$.
By \citet[Theorem 15.11]{AB2006}, $\Msc_1$ is weakly compact.\footnote{In \citet[Chapter 15]{AB2006}, the weak* topology refers to the weak topology here. Moreover, weak compactness (resp. closedness) here refers to the compactness (resp. closedness) under the weak topology.} 

For any $Q\in\Qsc$, let
$$A_Q(0)\trieq Q(1-),\;  A_Q(s)\trieq {1\over s}\int_{1-s}^1Q(t)dt,\quad s\in(0,1].$$
Obviously, $A_Q$ is decreasing and lower bounded on $[0,1]$. 
For any $Q\in\Qsc$, 
\begin{align*}\label{ineq:sQm}
Q\icx Q_0\Longleftrightarrow&\int_{[0,1]}A_Q(s) dm(s)\ge \int_{[0,1]}A_{Q_0}(s) dm(s),\quad\forall m\in\Msc\\
\Longleftrightarrow&\int_{[0,1]}A_Q(s) dm(s)\ge \int_{[0,1]}A_{Q_0}(s) dm(s),\quad\forall m\in\Msc_1.
\end{align*}
Similarly to \citet[Lemma 4.69]{FS2016},  the identity
\begin{equation}\label{eq:psi}
\begin{cases}
    w^\prime(t)=\displaystyle\int_{(1-t,1]}s^{-1} dm(s),\quad t\in[0,1)  \\
     w(1)-w(1-)=m(0)
\end{cases}
\end{equation}
defines a bijection
$$\Jc: \Msc\to\Wicx,$$ where $w^\prime$ denotes the right derivative. Under identity \eqref{eq:psi}, an application of Fubini theorem implies that
\begin{align*}
\int_{[0,1]}Q(t)dw(t)=\int_{[0,1]}A_Q(t)dm (t),\quad\forall\,Q\in \Qsc.
\end{align*}
Let
\begin{align*}
L_1(Q,m;\beta,\lambda)\trieq &\int_0^1(Q(s)-\beta)^2ds+\lambda\int_0^1Q(s)Q_\rho(1-s)ds\\
&-\int_{[0,1]}A_Q(s) dm(s)+ \int_{[0,1]}A_{Q_0}(s) dm(s),\quad
Q\in\Qsc,\; m\in\Msc.
\end{align*}
Then, under identity \eqref{eq:psi},
$$L(Q,w;\beta,\lambda)=L_1(Q,m;\beta,\lambda).$$
As a consequence,  Proposition \ref{prop:saddle:L} can be reformulated as the following one.

\begin{proposition}\label{prop:saddle:L1}
Assume that $X_0\in L^\infty$. Let $Q^*\in\Qsc$.
Then
 $Q^*$ solves problem \eqref{opt:lambda} if and only if there exists some $m^*\in\Msc$ such that $(Q^*,m^*)$ is a saddle point of $L_1(\cdot\;, \cdot\;; \beta,\lambda)$ (with respect to minimizing in $Q$ and maximizing in $m$).
\end{proposition}

Before the proof of Proposition \ref{prop:saddle:L1}, we introduce two lemmas. 

For every $\beta\in\Rbb$, $\lambda\ge0$, and $\mu\ge0$, let
\begin{align*}
K(Q,m,\mu;\beta,\lambda)\trieq &\int_0^1(Q(s)-\beta)^2ds+\lambda\int_0^1Q(s)Q_\rho(1-s)ds\\
&-\mu\int_{[0,1]}A_Q(s) dm(s)+ \mu\int_{[0,1]}A_{Q_0}(s) dm(s),\quad
Q\in\Qsc,\; m\in\Msc_1.
\end{align*}

\begin{lemma}\label{lma:K} 
Assume that $X_0\in L^\infty$.
For every $\beta\in\Rbb$, $\lambda\ge0$, and $\mu\ge0$, there exists some $m_1(\mu)\in\Msc_1$ such that
\begin{equation}\label{eq:infsup:K}
\begin{split}
\inf_{Q\in\Qsc}K(Q,m_1(\mu),\mu;\beta,\lambda)=&\sup_{m\in\Msc_1}\inf_{Q\in\Qsc}K(Q,m,\mu;\beta,\lambda)\\
=&\inf_{Q\in\Qsc}\sup_{m\in\Msc_1}K(Q,m,\mu;\beta,\lambda).
\end{split}
\end{equation}
\end{lemma}

\proof It is obvious that $K$ is convex in $Q$ and affine in
$m$. 
By the boundedness of $X_0$, we know that $A_{Q_0}$ is bounded and continuous on $[0,1]$.  Therefore, the functional
$$\Msc_1\ni m\mapsto \int_{[0,1]}A_{Q_0}(s) dm(s)$$
is weakly continuous.\footnote{Here the weak continuity refers to the continuity under the weak topology. Similarly, the weak upper (resp. lower) semi-continuity below refers to the
upper (resp. lower) semi-continuity under the weak topology.}
For each $n\ge 1$ and $Q\in\Qsc$, 
let $Q^n$ be given by
$$Q^n(s)=\min\{Q(s),n\},\quad s\in[0,1].$$
We know that, for each $n\ge1$, $A_{Q^n}$ is bounded and continuous on $[0,1]$ and hence 
the functional
$$\Msc_1\ni m\mapsto \int_{[0,1]}A_{Q^n}(s) dm(s)$$ 
is weakly continuous.
The monotone convergence theorem implies that
$$\int_{[0,1]}A_Q(t)dm(t)=\sup_{n\ge 1}\int_{[0,1]}A_{Q^n}(t)dm(t).$$
Therefore, for each $Q\in\Qsc$, the functional
$$\Msc_1\ni m\mapsto \int_{[0,1]}A_Q(t)dm(t)$$ 
is weakly lower semi-continuous.  As a consequence of the above discussion, $K$ is upper semi-continuous in $m$. 
By the weak compactness of $\Msc_1$ and by the minimax theorem (e.g.,
\citet[Proposition I.1.3]{MSZ2015}), there exists some
$m_1(\mu)\in\Msc_1$ that satisfies \eqref{eq:infsup:K}.
\qed

Obviously, problem \eqref{opt:lambda} is equivalent to the following one:
\begin{equation}\label{opt:lambda2}
\begin{split}
&\mini_{Q\in\Qsc}\quad \int_0^1(Q(s)-\beta)^2ds+\lambda\int_0^1Q(s)Q_\rho(1-s)ds\\
&\text{ subject to }\quad \inf_{m\in\Msc_1}\left(\int_{[0,1]}A_Q(s) dm(s)- \int_{[0,1]}A_{Q_0}(s) dm(s)\right)\ge0.
\end{split}
\end{equation}
Let
\begin{align*}
J(Q,\mu;\beta,\lambda)&\trieq \int_0^1(Q(s)-\beta)^2ds+\lambda\int_0^1Q(s)Q_\rho(1-s)ds\\
&-\mu\inf_{m\in\Msc_1}\left(\int_{[0,1]}A_Q(s) dm(s)- \int_{[0,1]}A_{Q_0}(s) dm(s)\right),\quad Q\in\Qsc,\mu\ge0.
\end{align*}

\begin{lemma}\label{lma:J}
Assume that $X_0\in L^\infty$.
Let $Q^*\in\Qsc$. Then $Q^*$
solves problem \eqref{opt:lambda2} if and only if there exists some
$\mu^*\ge0$ such that 
\begin{eqnarray}\label{ineq:saddle:J}
J(Q^*,\mu;\beta,\lambda)\le J(Q^*,\mu^*;\beta,\lambda)\le J(Q,\mu^*;\beta,\lambda)
\quad \forall\, Q\in\Qsc, \mu\ge0.
\end{eqnarray}
\end{lemma}

\proof Let $Q(s)=Q_0(s)+\alpha$ for all $s\in[0,1]$, where
$\alpha>0$. Obviously,  $Q\in\Qsc$ and
$$\inf_{m\in\Msc_1}\left(\int_{[0,1]}A_Q(s) dm(s)- \int_{[0,1]}A_{Q_0}(s) dm(s)\right)=\alpha>0.$$
Thus, the Slater condition for the constraint in \eqref{opt:lambda2} is satisfied. 
Then the conclusion of the lemma is an implication of \citet[Corollary 1 on p. 219 and Theorem 2 on p. 221]{L1969}. 
\qed

\paragraph{Proof of Proposition \ref{prop:saddle:L1}.} 
Obviously, $\Msc$ is the cone generated by $\Msc_1$. For all $Q\in\Qsc$, $m\in\Msc_1$, and $\mu\ge0$, we have
\begin{align*}
&K(Q,m,\mu;\beta,\lambda)=L_1(Q,\mu m;\beta,\lambda),\\
&J(Q,\mu;\beta,\lambda)=\sup_{m\in\Msc_1}K(Q,m,\mu;\beta,\lambda)=\sup_{m\in\Msc_1}L_1(Q,\mu m;\beta,\lambda).
\end{align*}
If $Q^*$ solves problem \eqref{opt:lambda}, or, equivalently,  problem \eqref{opt:lambda2}.
Let $\mu^*$ be given by Lemma \ref{lma:J} and $m_1(\mu^*)$ be given by Lemma \ref{lma:K}. Set $m^*=\mu^* m_1(\mu^*)$. Then by Lemma \ref{lma:fgh}, we know that $(Q^*,m^*)$ is a saddle point of $L_1(\cdot\;, \cdot\;; \beta,\lambda)$. Therefore, the ``only if'' part is proved. The ``if'' part is obvious.\qed

\subsection{Proof of Lemma \ref{lma:w(1)}}\label{app:proof:w(1)}
 Let $w\in\Wicx$. 
 The boundedness of $X_0$ implies that $\int_{[0,1]}Q_0(s)dw(s)\in(-\infty,\infty)$
 Assume that $w(1)-w(1-)>0$. For any $\varepsilon\in(0,1)$, let $Q^\varepsilon\in\Qsc$ be given by
$$Q^\varepsilon(s)={1\over\sqrt{\varepsilon}}\id_{s\ge 1-\varepsilon}, \quad s\in[0,1].$$ 
Then we have
\begin{align*}
&L(Q^\varepsilon,w;\beta,\lambda)-\int_{[0,1]}Q_0(s)dw(s)\\
=&1-2\beta\sqrt{\varepsilon}+\beta^2+{\lambda\over\sqrt{\varepsilon}}\int_{1-\varepsilon}^1Q_\rho(1-s)ds-{w(1)-w(1-\varepsilon)\over\sqrt{\varepsilon}}\\
\le&1-2\beta\sqrt{\varepsilon}+\beta^2+{\lambda\sqrt{\varepsilon}}Q_\rho(\varepsilon)-{w(1)-w(1-)\over\sqrt{\varepsilon}}
\overset{\varepsilon\downarrow0}\longrightarrow -\infty,
\end{align*} 
which leads to $\inf_{Q\in \Qsc} L(Q,w;\beta,\lambda)=-\infty$. \qed

\subsection{Proof of Lemma \ref{lma:inner:infinite}}\label{app:prrof:inner:infinite}
Assume that $w(1)=w(1-)$ and $w^\prime\notin L^2([0,1))$. Then $\int_0^1(w^\prime(s))^2ds=\infty$. For each $n\ge1$, let $Q_n$ be given by
$$Q_n(s)=\beta+{f_n(s)-\lambda Q_\rho((1-s)-)\over2},\quad s\in[0,1),$$
where
$$f_n(s)=\min\{w^\prime(s),n\},\quad s\in[0,1).$$ 
Obviously, $Q_n\in\Qsc$ for all $n\ge1$ and 
$$\int_0^1f_n^2(s)ds\to\int_0^1(w^\prime(s))^2ds=\infty\text{ as }n\to\infty.$$
The boundedness of $X_0$ implies that $a_1\trieq \int_0^1Q_0(s)w^\prime(s)ds\in(-\infty,\infty)$. Moreover,
\begin{align*}
a_2\trieq & \int_0^1Q_\rho(1-s)w^\prime(s)ds\\
\le & w^\prime\left({1\over2}\right)\int_0^{1\over2}Q_\rho(1-s)ds+Q_\rho\left({1\over2}\right)\left(w(1)-w\left({1\over2}\right)\right)<\infty.
\end{align*}
By $f_n\le w^\prime$, we have
$\int_0^1Q_\rho(1-s)f_n(s)ds\le a_2$. 
Then 
\begin{align*}
&L(Q_n,w;\beta,\lambda)\\
=&\int_0^1{(f_n(s)-\lambda Q_\rho(1-s))^2\over 4}ds+\lambda\int_0^1\left(\beta+{f_n(s)-\lambda Q_\rho(1-s)\over 2}\right)Q_\rho(1-s)ds\\
&-\int_0^1\left(\beta+{f_n(s)-\lambda Q_\rho(1-s)\over 2}\right)w^\prime(s)ds+a_1\\
\le&\int_0^1{(f_n(s)-\lambda Q_\rho(1-s))^2\over 4}ds+\lambda\beta\ep[\rho]+{\lambda a_2\over2}\\
&\quad -\beta w(1)-{1\over2}\int_0^1f_n(s)w^\prime(s)ds+{\lambda a_2\over2}+a_1\\
\le &{1\over4}\int_0^1f_n^2(s)ds+{\lambda^2\over4}\ep[\rho^2]+\lambda\beta\ep[\rho]
+\lambda a_2-\beta w(1)+ a_1-{1\over2}\int_0^1f_n(s)w^\prime(s)ds\\
\le&{1\over4}\int_0^1f_n^2(s)ds+{\lambda^2\over4}\ep[\rho^2]+\lambda\beta\ep[\rho]
+\lambda a_2-\beta w(1)+ a_1-{1\over2}\int_0^1f^2_n(s)ds\\
=&-{1\over4}\int_0^1f_n^2(s)ds+{\lambda^2\over4}\ep[\rho^2]+\lambda\beta\ep[\rho]
+\lambda a_2-\beta w(1)+a_1\\
\to& -\infty \quad\text{ as }n\to\infty.
\end{align*}
Therefore, $\inf_{Q\in \Qsc} L(Q,w;\beta,\lambda)=-\infty$. 
\qed

\subsection{Proof of Proposition \ref{prop:2point}}\label{app:prop;2point}

For every $\lambda>0$,
$$N^\prime_\lambda(s)=
\begin{cases}
\lambda Q_\rho((1-s)-)+2a &\text{if }s\in[0,p),\\
\lambda Q_\rho((1-s)-)+2b &\text{if }s\in[p,1).\\
\end{cases}
$$
Obviously, $N_\lambda$ is concave on each of the intervals $[0,p)$ and $[p,1]$.
Therefore, $\breve N^\prime_\lambda$ is constant on each of the two intervals.

We divide the discussion into the following two cases.

\begin{description}

\item[(i)] Assume that
$\lambda\ge {2b-2a\over A_1-A_2}$.

In this case, 
${N_\lambda(p)\over p}\ge {N_\lambda(1)-N_\lambda(p)\over 1-p}$.
Then
\begin{align*}
\breve N^\prime_\lambda(s)=N_\lambda(1)=\lambda\ep[\rho]+2\ep[X_0],\quad s\in[0,1).
\end{align*}
Therefore,
\begin{align*}
h(\beta,\lambda)=
\begin{cases}
2\beta-2\ep[X_0]\quad&\text{if }2\beta\le \lambda\ep[\rho]+2\ep[X_0],\\
\lambda\ep[\rho]\quad&\text{if }2\beta> \lambda\ep[\rho]+2\ep[X_0].
\end{cases}
\end{align*}
As a consequence, $\beta_\lambda=\ep[X_0]$,
$$Q^*_\lambda(s)=\ep[X_0]+{\lambda\over2}\ep[\rho]-{\lambda\over2}Q_\rho((1-s)-),\quad s\in[0,1)$$
and
\begin{align*}
\Xc(\lambda)=\ep[X_0]\ep[\rho]-{\lambda\over2}\var[\rho].
\end{align*}

\item[(ii)] Assume that $\lambda< {2b-2a\over A_1-A_2}$.

In this case, 
${N_\lambda(p)\over p}< {N_\lambda(1)-N_\lambda(p)\over 1-p}$.
Then
\begin{align*}
\breve N^\prime_\lambda(s)=&
\begin{cases}
{N_\lambda(p)\over p}&\text{if }  s\in[0,p)\\
{N_\lambda(1)-N_\lambda(p)\over 1-p}&\text{if }  s\in[p,1)
\end{cases}\\
=&
\begin{cases}
2a+\lambda A_1&\text{if } s\in[0,p),\\
2b+\lambda A_2&\text{if } s\in[p,1).
\end{cases}
\end{align*}
Therefore,
\begin{align*}
h(\beta,\lambda)=
\begin{cases}
2\beta-2\ep[X_0]& \text{if }2\beta\le 2a+\lambda A_1,\\
2\beta (1-p)+\lambda p A_1-2b(1-p)& \text{if }2a+\lambda A_1< 2\beta\le 2b+\lambda A_2,\\
\lambda\ep[\rho] & \text{if }2\beta>2b+\lambda A_2.
\end{cases}
\end{align*}
As a consequence, 
\begin{align*}
\beta_\lambda=
\begin{cases}
\ep[X_0]
& \text{if }\lambda\ge {(2b-2a)(1-p)\over A_1},\\
b-{\lambda p\over 2(1-p)}A_1& \text{if } \lambda< {(2b-2a)(1-p)\over A_1}.
\end{cases}
\end{align*}
Moreover,  if $$\lambda\ge {(2b-2a)(1-p)\over A_1}
,$$
then
$$Q^*_\lambda(s)=
\begin{cases}
a+{\lambda\over 2}A_1-{\lambda\over2}Q_\rho((1-s)-)&\text{if }s\in[0,p)\\
b+{\lambda\over 2}A_2-{\lambda\over2}Q_\rho((1-s)-)&\text{if }s\in[p,1)
\end{cases}
$$ 
and hence
\begin{align*}
\Xc(\lambda)=ap A_1+b(1-p)A_2+{\lambda p\over2}A_1^2+{\lambda(1-p)\over 2}A_2^2-{\lambda\over2}\ep[\rho^2];
\end{align*}
if $$\lambda< {(2b-2a)(1-p)\over A_1},$$
then
$$Q^*_\lambda(s)=
\begin{cases}
b-{\lambda p\over 2(1-p)}A_1-{\lambda\over2}Q_\rho((1-s)-)&\text{if }s\in[0,p)\\
b+{\lambda\over 2}A_2-{\lambda\over2}Q_\rho((1-s)-)&\text{if }s\in[p,1)
\end{cases}
$$
and hence
\begin{align*}
\Xc(\lambda)
=b\ep[\rho]+{\lambda\over2(1-p)}\left((1-p)^2A_2^2-p^2A_1^2\right)-{\lambda\over2}\ep[\rho^2].
\end{align*}

\end{description}

Summarizing the discussions in the above two cases,
we have
\begin{align*}
\Xc(\lambda)=
\begin{cases}
\ep[X_0]\ep[\rho]-{\lambda\over2}\var[\rho] 
\quad\text{if }\lambda\ge {2b-2a\over A_1-A_2},\\
ap A_1+b(1-p)A_2+{\lambda p\over2}A_1^2+{\lambda(1-p)\over 2}A_2^2-{\lambda\over2}\ep[\rho^2]\\
\hskip5cm \text{if } {(2b-2a)(1-p)\over A_1}\le \lambda<{2b-2a\over A_1-A_2},\\
b\ep[\rho]+{\lambda\over2(1-p)}\left((1-p)^2A_2^2-p^2A_1^2\right)-{\lambda\over2}\ep[\rho^2]
\quad\text{if }0<\lambda< {(2b-2a)(1-p)\over A_1}.
\end{cases}
\end{align*}
Then solving equation $\Xc(\lambda)=x$ yields  
\begin{align*}
\lambda^\circ=
\begin{cases}
{2(\ep[X_0]\ep[\rho]-x)\over\var[\rho]}\quad\text{if }x\le \ep[X_0]\ep[\rho]-{b-a\over A_1-A_2}\var[\rho],\\
{2(ap A_1+b(1-p)A_2-x)\over\ep[\rho^2]-pA_1^2-(1-p)A_2^2}\\
\quad\text{if }
\ep[X_0]\ep[\rho]-{b-a\over A_1-A_2}\var[\rho]< x\\
\hskip3cm \le b\ep[\rho]-{b-a\over A_1}\left((1-p)\ep[\rho^2]+p^2A_1^2-(1-p)^2A_2^2\right),
\\
{2(1-p)(b\ep[\rho]-x)\over(1-p)\ep[\rho^2]+p^2A_1^2-(1-p)^2A_2^2}\\
\quad\text{if }b\ep[\rho]-{b-a\over A_1}\left((1-p)\ep[\rho^2]+p^2A_1^2-(1-p)^2A_2^2\right)<x <b\ep\rho].
\end{cases}
\end{align*}
Finally, a substitution leads to Proposition \ref{prop:2point}. \qed

\end{appendices}


\bibsep=0pt
\bibliographystyle{abbrvnat}
\bibliography{bicxbib}

\end{document}